\documentclass[10pt]{article}
\usepackage{array}
\usepackage{faktor}
\usepackage{wrapfig, amsmath}
\usepackage{amsfonts}
\usepackage{amssymb}
\usepackage{bm}
\usepackage{graphicx}
\usepackage{epsfig,cancel,amsthm, amssymb, todonotes, amsmath}
\usepackage{color,tikz}
\usetikzlibrary{decorations.markings,arrows, calc}
\usepackage{color}
\usepackage{graphicx,framed,verbatim,caption,amsbsy}
\usepackage[colorlinks=true, pdfstartview=FitV, linkcolor=darkblue, citecolor=darkblue, urlcolor=darkblue]{hyperref}
\usepackage{enumitem}
\usepackage[rflt]{floatflt}

\definecolor{shadecolor}{rgb}{0.9, 0.9, 0.86}
\definecolor{darkgreen}{rgb}{0.2, 0.5,  0}
\definecolor{darkblue}{rgb}{0.1,0.1,0.45}

\def\&{\vspace{-5pt}&}

\def\Re{\mathrm {Re}\,}
\def\Im{\mathrm {Im}\,}

\usepackage{tikz}
\usepackage{pgf}
\usepackage{pgflibraryarrows}
\usepackage{pgflibrarysnakes}
\usetikzlibrary{decorations.text}
\usepgfmodule{shapes}
\usetikzlibrary{decorations.pathmorphing}
\usetikzlibrary{decorations.markings}
\usetikzlibrary{patterns}
\usetikzlibrary{automata}
\usetikzlibrary{positioning}

\tikzset{->-/.style={decoration={
 markings,
 mark=at position #1 with {\arrow{>}}},postaction={decorate}}}

\def \eqref#1{(\ref{#1})}
\def \& {&\hspace{-10pt}}
\def\Ai{ {\mathrm {Ai}}}

\newcommand{\G}{\Gamma} 
\renewcommand{\d}{\mathrm d}
       
\newtheorem{theorem}{Theorem}[section]
\newtheorem{example}[theorem]{Example}
\newtheorem{exercise}[theorem]{Exercise}

\newtheorem{lemma}[theorem]{Lemma}
\newtheorem{remark}[theorem]{Remark}

\newtheorem{proposition}[theorem]{Proposition} 
\newtheorem{corollary}[theorem]{Corollary} 
 
\newtheorem{definition}[theorem]{Definition}

\def\le{\left}
\def\ri{\right}
\def\d{{\rm d}}
\def\ds{\displaystyle}

\def\res{\mathop{\mathrm {res}}\limits_}

\def\bt{\begin{theorem}}
\def\et{\end{theorem}}
\def\bc{\begin{corollary}}
\def\ec{\end{corollary}}
\def\bx{\begin{example}}
\def\ex{\end{example}}
\def\bxr{\begin{exercise}\small}
\def\exr{\end{exercise}}
\def\bl{\begin{lemma}}
\def\el{\end{lemma}}
\def\bd{\begin{definition}}
\def\ed{\end{definition}}
\def\bp{\begin{proposition}}
\def\ep{\end{proposition}}

\def\br{\begin{remark}}
\def\er{\end{remark}}

\def\be{\begin{equation}}
\def\ee{\end{equation}}

\def\&{\hspace{-15pt}&}
\def\bea{\begin{eqnarray}}
\def\eea{\end{eqnarray}}
\def\beas{\begin{eqnarray*}}
\def\eeas{\end{eqnarray*}}

\def\R{{\mathbb R}}

\def\1{{\bf 1}}

\def\Id{{\mathrm{Id}}}

\def\Id{\mathbf{I}}

\def\F{\mathcal{F}}
\def\G{\mathcal{G}}

\def\QED {\hfill $\blacksquare$\par\vskip 3pt}

\renewcommand{\theequation}{\arabic{section}.\arabic{equation}}

\makeatletter
\@addtoreset{equation}{section}
\makeatother

 \usepackage{authblk}

\date{}                     

\title{Fredholm determinant solutions of the Painlev\'e II hierarchy and gap probabilities of determinantal point processes}
\author[1]{Mattia Cafasso}
\author[2]{Tom Claeys}
\author[3]{Manuela Girotti} 
\affil[1]{\textit{LAREMA - Universit\'e d'Angers, 2 Boulevard Lavoisier, 49045 Angers, France;} \texttt{cafasso@math.univ-angers.fr}}
\affil[2]{\textit{Institut de Recherche en Math\'ematique et Physique,  UCLouvain, Chemin du Cyclotron 2, B-1348 Louvain-La-Neuve, Belgium;} \texttt{tom.claeys@uclouvain.be}}
\affil[3]{\textit{Department of Mathematics, Colorado State University, 1874 campus delivery, Fort Collins, CO 80523 and Department of Mathematics and Statistics, Concordia University, Montr\'eal, QC;} \texttt{manuela.girotti@concordia.ca}}
\begin{document}
\maketitle

\begin{abstract}
We study Fredholm determinants of a class of integral operators, whose kernels can be expressed as double contour integrals of a special type. Such Fredholm determinants appear in various random matrix and statistical physics models. We show that the logarithmic derivatives of the Fredholm determinants are directly related to solutions of the Painlev\'e II hierarchy. This confirms and generalizes a recent conjecture by Le Doussal, Majumdar, and Schehr \cite{FermMM}. In addition, we obtain asymptotics at $\pm\infty$ for the Painlev\'e transcendents and large gap asymptotics for the corresponding point processes.
\end{abstract}




\section{Introduction}
We consider a family of integral kernels $K(x,y)$, defined for $x,y\in\mathbb R$, of the form
\begin{equation}\label{def:kernel}
K(x,y)={\dfrac{1}{(2 \pi i)^2}}\int_{\gamma_R} \d\mu \int_{\gamma_L} \d\lambda \frac{{\rm e}^{(-1)^{n+1}(p_{2n+1}(\mu)-p_{2n+1}(\lambda))-x\mu+y\lambda}}{\lambda - \mu},
\end{equation}
where $n\in\mathbb N$, and $\gamma_L$ and $\gamma_R$ are curves in the left and right half of the complex plane without self-intersections which are asymptotic to straight lines with arguments $\pm \frac{n}{2n+1}\pi$ at infinity, as illustrated in Figure \ref{Fig1}. The function $p_{2n+1}$ can be any odd polynomial of the form
\begin{equation}\label{def p}
p_{2n+1}(x)=\frac{x^{2n+1}}{2n+1}+
\sum_{j=1}^{n-1}\frac{
        \tau_{j}}{2j+1}x^{2j+1},
    \end{equation}
    with $\tau_1,\ldots, \tau_{n-1}\in\mathbb R$.
Determinantal point processes with kernels of this form appear in various random matrix and statistical physics models. The most prominent example is the Airy kernel, which corresponds to $n=1$ and $p_{3}(\lambda)=\lambda^3/3$. A generalized Airy point process characterized by the kernel \eqref{def:kernel} with monomial $p_{2n+1}$ appeared recently in \cite{FermMM}, as a model for the momenta of a large number of fermions trapped in a non--harmonic potential. Point processes with kernels of the general form \eqref{def:kernel}, but with $p_{2n+1}$ not necessarily polynomial, also appear, for instance, in \cite{Johansson}.
In \cite{ACvM1}, a connection was established between the kernels \eqref{def:kernel}, the theory of Gelfand--Dickey equations and some specific solutions which are related to topological minimal models of type $A_n$ (see \cite{Dij} for an excellent introduction to this subject). The monomial case for $n = 2$ is also related to spin chains in a magnetic field \cite{Stephan}.\\

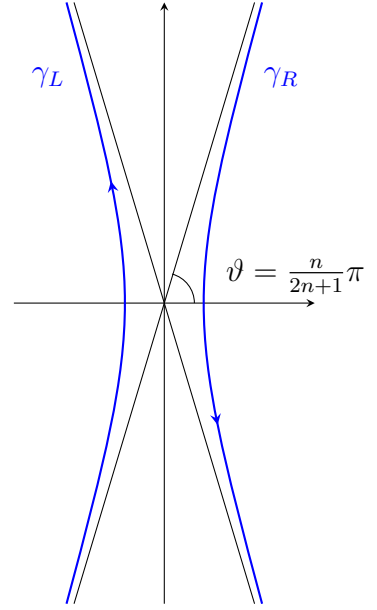
\begin{floatingfigure}[r]{.3\textwidth}
\centering
\scalebox{1}{
\begin{tikzpicture}[>=stealth]
\path (0,0) coordinate (O);

\draw[->] (-2,0) -- (2,0) coordinate (x axis);
\draw[->] (0,-4) -- (0,4) coordinate (y axis);

\draw (1.2,4) -- (-1.2,-4);
\draw (1.2,-4) -- (-1.2,4);
\draw (.4,0) arc (0:75:0.4);
\node [above right] at (.7,0) {\large $\vartheta = \frac{n}{2n+1}\pi$};

\draw[thick, blue, ->- = .7]  (-1.3,-4) .. controls + (75:4cm) and + (-75:4cm) .. (-1.3,4);
\node [left] at (-1.2,3) {{\color{blue} \large $ \gamma_L$}};

\draw[thick, blue, ->- = .7]  (1.3,4) .. controls + (-105:4cm) and + (105:4cm) .. (1.3,-4);
\node [right] at (1.2,3) {{\color{blue} \large $ \gamma_R$}};
\end{tikzpicture}}
\caption{The contours $\gamma_R$ and $\gamma_L$.  \label{Fig1}}
\end{floatingfigure}

Our objective is to derive properties of the Fredholm determinants
	\begin{multline}\label{def:Fredholm}
		F(s;\varrho):=\det(\Id-\left. \varrho\, \mathcal{K}\right|_{[s,+\infty)})\\
		=1 + \sum_{k = 1}^\infty \frac{(-\varrho)^k}{k!}  \int_{[s,+\infty)^k} \det \Big[K(\xi_i,\xi_j)\Big]_{i,j = 1}^k\prod_{i = 1}^k \d\xi_i,
	\end{multline}
for $s\in\mathbb R$ and $0<\varrho\leq 1${, where $\left.\mathcal K\right|_{[s,+\infty)}$ is the integral operator with kernel $K$ acting on $L^2(s,+\infty)$.}
Throughout the paper, in order to distinguish integral kernel operators from their kernels, we will denote integral kernel operators by calligraphic letters ($\mathcal K$ in the above equation) and the associated kernels by the corresponding uppercase letter ($K$ in the above equation).\\ 
Given a determinantal point process on the real line characterized by a correlation kernel $K$, the Fredholm determinant $F(s;1)$ is the probability distribution of the largest particle $\zeta_{\rm max}$ (which exists almost surely if the Fredholm determinant is well-defined) in the process, $F(s;1)=\mathbb P(\zeta_{\rm max}<s)$, see e.g.\ \cite{Johansson, Sosh}. For $\varrho\in (0,1]$, $F(s;\varrho)=\mathbb P(\zeta_{\rm max}^{(\varrho)}<s)$ is the probability distribution of the largest particle $\zeta_{\rm max}^{(\varrho)}$ in the 
associated thinned process, which is obtained from the original process by removing each of the particles independently with probability $1-\varrho$, see e.g.\ \cite{BohigasPato,BothnerBuckingham, BDIK, ClaeysDoeraene}.
In Appendix \ref{appendix:genAiry}, we confirm using standard methods that the kernels \eqref{def:kernel} indeed define a point process for any choice of $n, \tau_1,\ldots, \tau_{n-1}$, and this implies, in particular, {that {$F(s;\varrho)$} is a distribution function for any $0< \varrho\leq 1$.

As our first result, we will prove that $F(s;\varrho)$ can be expressed explicitly in terms of solutions to the Painlev\'e II hierarchy.

\medskip

The Painlev\'e II hierarchy is a sequence of ordinary differential equations obtained from the equations of the mKdV hierarchy via self-similar reduction \cite{FN}. Using the same normalization as in \cite{CIK}, the $n$-th member of the Painlev\'e II hierarchy is an equation for $q=q(s)$ defined as follows:
\begin{equation}\label{PIIn}
\left(\frac{\d}{\d s}+2q\right)\mathcal L_n[q_s-q^2]+\sum_{\ell=1}^{n-1}\tau_{\ell} \left(\frac{\d}{\d s}+2q\right)\mathcal L_\ell[q_s-q^2]=sq-\alpha, \qquad n\geq 1,
\end{equation}
where $\alpha, \tau_1,\ldots, \tau_{n-1}$ are real parameters, and where the operators $\{\mathcal L_n, \; n \geq 0\}$ are the Lenard operators defined recursively by
\begin{equation}\label{def L}
\frac{\d}{\d s}\mathcal L_{j+1}f=\left(\frac{\d^3}{\d s^3}+4f\frac{\d}{\d s}+2f_s\right)\mathcal L_j f, \qquad \mathcal L_0 f=\frac{1}{2},\qquad \mathcal L_j 1=0,\quad j\ge 1.
\end{equation}
The first members of the hierarchy are\footnote{The function $q$ corresponds to the function $g$ in \cite{FermMM} if we set the parameters $\tau_i$ to zero.}
\begin{align}
& q''-2q^3=sq-\alpha, \nonumber\\
&q''''-10 q(q')^2-10q^2q''+6q^5+\tau_1(q''-2q^3)=sq-\alpha, \nonumber\\
&q''''''-14q^2q''''-56qq'q'''-70(q')^2q''-42q(q'')^2
+70q^4q''+140q^3(q')^2-20q^7 \nonumber \\&\qquad\qquad\qquad\qquad\qquad\qquad+\tau_2(q''''-10 q(q')^2-10q^2q''+6q^5)
+\tau_1(q''-2q^3)=sq-\alpha, \label{PIIhierarchy}
\end{align}
where $'$ stands for $\dfrac{\d}{\d s}$. 
Generally speaking, even if there exist families of rational and special function solutions to Painlev\'e equations and hierarchies,
one can say that typical solutions are transcendental functions which have no simple closed expression. 
We construct a family of solutions to the Painlev\'e II hierarchy in terms of the Fredholm determinants \eqref{def:Fredholm}.
This is also of interest from a numerical point of view: whereas it is in general a challenge to accurately evaluate Painlev\'e transcendents numerically, there are efficient algorithms to compute Fredholm determinants \cite{Bornemann}.\\

We will be interested in the homogeneous version of the Painlev\'e II hierarchy, which corresponds to $\alpha=0$. The family of solutions that we will construct, contains natural generalizations of the Hastings-McLeod solution \cite{HastingsMcLeod} and the Ablowitz-Segur solutions \cite{AblowitzSegur} to the (second order) Painlev\'e II equation.
These solutions behave like a multiple of the Airy function as $s\to +\infty$.
More precisely, for every $n$, we construct solutions $q(s;\varrho)$ of the $n$-th member of the hierarchy in terms of $F(s;\varrho)$, in such a way that $q\big((-1)^{n+1}s;\varrho\big)$ decays rapidly at $+\infty$, and behaves like a root function at $-\infty$ if $\varrho=1$.
\begin{theorem}\label{thm:PIIhier}
	Let $n\in\mathbb N$, {$0<\varrho\leq 1$} $\tau_1,\ldots, \tau_{n-1}\in\mathbb R$, and let $F(s;\varrho)=F(s;\varrho;\tau_1,\ldots, \tau_{n-1})$ be the Fredholm determinant defined in 	\eqref{def:Fredholm} with $K$ given by \eqref{def:kernel}--\eqref{def p}.
There is a real solution $q(s;\varrho)=q(s;\varrho;\tau_1,\ldots, \tau_{n-1})$ to the equation of order $2n$ in the Painlev\'e II hierarchy \eqref{PIIn} which has no poles for real $s$, which satisfies
\be\label{eqthm1}
q^2\big((-1)^{n+1}s;\varrho\big)=-\frac{\d^2}{\d s^2}\log F(s;\varrho),
\ee
and which has the asymptotic behavior
\begin{equation}
q({(-1)^{n+1}}s;\varrho)=\mathcal{O}\le(e^{-C s^{\frac{2 n + 1}{2 n}}}\ri),\qquad \mbox{ as $s\to +\infty$,}\label{eq:qas+0}\end{equation}
for some $C>0$.
Moreover, if $\varrho=1$, 
\begin{equation}
	q((-1)^{n+1}s;1)\sim \left(\frac{n!^2}{(2n)!}|s|\right)^{\frac{1}{2n}},\qquad \mbox{ as $s\to -\infty$.}\label{eq:qas-}
\end{equation}
For any {$0<\varrho\leq 1$}, we also have the identity
\begin{equation}\label{eq:TWidentity}
F(s;\varrho) =\exp\left\{-\int_{s}^{+\infty}(x-s)\, q^2\big({(-1)^{n+1}}x;\varrho\big)\, \d x\right\},
\end{equation}
for $F$ in terms of $q$.\\
\end{theorem}
In \cite{FermMMarxiv}, the authors studied $F(s;1)$  for the case of monomial $p_{2n + 1}$. Generalizing the approach used by Tracy and Widom in \cite{TracyWidomLevel}, they proved that $F(s;1)$ has an explicit expression in terms of a particular solution of the quite simple Hamiltonian system of dimension $2n$ :\begin{equation}\label{Hamsystem} q'_p = q_{p+1} - q_0u_p, \quad u'_p = -q_0q_p, \quad 0\leq p\leq 2n - 1.\end{equation} Namely, $\frac{\d^2}{\d s^2}\log F(s;1) = -q_0^2(s)$. Then, they conjectured that this Hamiltonian system implies that $q_0$ solves the equation of order $2n$ of the Painlev\'e II hierarchy, where all the parameters $\tau_1,\ldots,\tau_{n-1}$ are set to $0$. The conjecture had been verified by the authors up to large $n$, by explicit computations. Theorem \ref{thm:PIIhier} proves an extended version of this conjecture connecting directly the gap probability of these processes with the Painlev\'e II hierarchy, without using the Hamiltonian system. It would still be interesting, as suggested by the authors of \cite{FermMMarxiv}, to find another proof of their conjecture comparing directly \eqref{Hamsystem} with the well known Hamiltonian system classically associated to the Painlev\'e II hierarchy, see \cite{MazzoccoMo}.
\begin{remark}
More detailed asymptotics for $q((-1)^{n+1}s;\varrho)$ as $s\to +\infty$ can in fact be predicted with the following heuristic arguments.
If $q\big((-1)^{n+1}s;\varrho\big)$ decays rapidly as $s\to +\infty$, it can be expected that the term $\frac{\d^{2n}}{\d s^{2n}}q$ at the left and the term $sq$ at the right are dominant in the equation \eqref{PIIn} for large $s$. The equation then reduces formally to the {generalized Airy} equation \[\frac{\d^{2n}}{\d s^{2n}} f(s) = s f (s),\] and it can therefore be expected that $q\big((-1)^{n+1}s\big)$ behaves for large $s$ like a real decaying solution of this equation.
The function $f(s) = \sqrt{\varrho}\, \Ai_{2n + 1}\big((-1)^{n+1}s\big)$ with
\be\label{def:genAiry}
\Ai_{2n + 1}(s) := (-1)^{n+1} \int_{\gamma_L} \frac{\d \lambda}{2 \pi i}\, {\rm e}^{(-1)^n\frac{\lambda^{2n + 1}}{2n + 1} + s \lambda}, 
\ee
is such a solution for any $\varrho$, see e.g. \cite{FermMM}.
\end{remark}

When all the parameters $\tau_1,\ldots,\tau_{n-1}$ are set to $0$, we confirm the above heuristics by proving that, with $q$ as in Theorem \ref{thm:PIIhier},
\be
	q\big((-1)^{n+1}s;\varrho \big)\sim \sqrt{\varrho}\, \Ai_{2n+1}(s), \quad \text{as} \; s\rightarrow \infty
\ee
(see equation \eqref{eq:asqAiry} below).
As explained in Remark \ref{remark2as}, proving this rigorously for general $\tau_1,\ldots,\tau_{n-1}$ would require a rather technical saddle point analysis, and we therefore do not pursue this.

\begin{remark}
For $n=1$, \eqref{eq:TWidentity} is the well-known Tracy-Widom distribution \cite{TracyWidomLevel}, which is among others, for $\varrho=1$, the limit distribution for extreme eigenvalues in many random matrix ensembles. For general $n$ {and $\varrho = 1$}, \eqref{eq:TWidentity} has the same structure as generalizations of the Tracy-Widom distributions obtained in \cite{CIK}, see also \cite{AkemannAtkin}, 
which describe the limit distribution for the extreme eigenvalues in unitary random matrix models with critical edge points.
However, we emphasize that the relevant solutions in \cite{CIK} are different from the ones present here, since they correspond to $\alpha=1/2$ and have different asymptotics. Therefore, the distribution functions appearing in \cite{CIK} are different from $F(s;1)$. 
\end{remark}

One cannot evaluate the Fredholm determinants $F(s;\varrho)$ explicitly for a given value of $s$, and therefore it is natural to try to approximate them for large values of $s$. In general, the $s\to +\infty$ asymptotics can be deduced directly from the asymptotics of the kernel $K(x,y)$ as $x,y\to +\infty$, but the $s\to -\infty$ are more delicate, and are commonly referred to as {\em large gap asymptotics} \cite{Forrester, Krasovsky}.
In our final result, we obtain $s\to -\infty$ asymptotics for $F(s;1)$ up to the value of a multiplicative constant. For the sake of clarity, let us first state the result in the simpler case where the parameters $\tau_1,\ldots,\tau_{n-1}$ are set to $0$.

\begin{theorem}\label{thm:largegaptau0}

Let $n\in\mathbb N$ and let $F(s;1)$ be the Fredholm determinant defined in \eqref{def:Fredholm} with $K$ given by \eqref{def:kernel} and monomial $p_{2n + 1}(x) := \frac{x^{2n + 1}}{2n + 1}$. 
As $s\to -\infty$, there exists a constant $C>0$, possibly depending on $n$, such that we have the asymptotics 
\begin{equation}\label{confirmation}
 F(s;1)=C |s|^{c}{\rm e}^{-\frac{n^2}{(n+1)(2n + 1)}{2n \choose n}^{-\frac{1}{n}}|s|^{2 + \frac{1}n}}\left(1+o(1)\right),
\end{equation}
{with $c=-\frac{1}{8}$ if $n=1$ and $c=-\frac{1}{2}$ otherwise.}
Moreover, the asymptotics \eqref{eq:qas-} can be improved to 
\begin{equation}
q({(-1)^{n+1}}s;1)=\Bigg(\frac{n!^2}{(2n)!}|s| \Bigg)^\frac{1}{2n}+\frac{c}{2}\Bigg(\frac{(2n)!}{n!^2} \Bigg)^{\frac{1}{2n}}|s|^{-2 - \frac{1}{2n}}+\mathcal O\left(|s|^{-2-\frac{1}{n}}\right) ,\qquad \mbox{ as $s\to -\infty$.}\label{eq:qas-2bis}
\end{equation}
\end{theorem}
\begin{remark}The leading order term of $\log F(s;1)$ as $s\to -\infty$, namely $-\frac{n^2}{(n+1)(2n + 1)}{2n \choose n}^{-\frac{1}{n}}|s|^{2 + \frac{1}n}$, was already predicted in \cite{FermMM}, but without any information about the subleading terms.
Our approach to prove Theorem \ref{thm:largegaptau0} consists of first deriving asymptotics for the logarithmic derivative
$\frac{\d}{\d s}\log F(s;1)$, which we then integrate in $s$. A consequence is that the multiplicative constant $C$ in \eqref{confirmation} arises as a constant of integration, which we are not able to evaluate. Explicitly evaluating such multiplicative constants in large gap asymptotics of Fredholm determinants is in general a difficult task, see \cite{Krasovsky}. In the Airy case $n=1$ with $\varrho=1$, it was proved in \cite{DIK, BBdF} that $C=\rm e^{\frac{1}{24}\log 2+\zeta'(-1)}$, where $\zeta'$ is the derivative of the Riemann $\zeta$ function, by approximating the Fredholm determinant by a Toeplitz or Hankel determinant for which the corresponding constant can be obtained via the evaluation of a Selberg integral. It seems unlikely that a similar approach could lead to the evaluation of $C$ for $n>1$.
\end{remark}

{In the more general case where the parameters $\tau_1,\ldots,\tau_{n-1}$ are different from zero, the {asymptotics} of $F(s;1)=F(s;1;\tau_1,\ldots, \tau_{n-1})$ {change} in a rather subtle way; we refer to {Section} \ref{sec:4} for {further} details. 
Let 
$$ \lambda(z) := \sum_{k=1}^{n}(-1)^{n-k} {2k \choose k} \tau_{k} z^{2k},\qquad \tilde\lambda(z):=\sum_{k=1}^{n}(-1)^{n-k} {2k \choose k} \tau_{k} z^{k},$$
and define $\theta_1,\ldots, \theta_{2n}$ and $\theta_1^{[2]},\ldots, \theta_{2n}^{[2]}$ as follows:
\begin{equation}
\theta_i(\tau_1,\ldots,\tau_n) \equiv \theta_i := \begin{cases} \ds{2n \choose n}^{-\frac{1}{2n}} \quad \mathrm{for}\; i=0,\\[2ex]										 \dfrac{1}{2i - 1} \res {z = \infty} \lambda^{\frac{2i-1}{2n}}(z), \quad \mathrm{for}\; i \geq 1, 										 \end{cases}\label{thetajintro}
\end{equation}
	where the residue at infinity is minus the coefficient of $z^{-1}$ in the large $z$ expansion of the branch of $\lambda^{\frac{2i-1}{2n}}(z)$
	which is positive for large $z>0$, and similarly
\begin{equation}
\theta_i^{[2]}(\tau_1,\ldots,\tau_n) \equiv \theta_i^{[2]} := \begin{cases} \ds\theta_0^2,&i=0,\\[2ex]
														\ds\frac{\ds\tau_{n - 1}}{4n-2},&i = 1,\\[2ex]

\ds\dfrac{1}{i - 1} \res {z = \infty} \tilde\lambda^{\frac{i-1}{n}}(z), &i \geq 2. 							\end{cases}\label{formula:thetaj2intro}\end{equation}	
We will explain in Section \ref{sec:4} that these numbers are related to topological minimal models of type $A_n$ and, more specifically, to flat coordinates for the corresponding Frobenius manifolds. One can also compute $\{ \theta_j^{[2]}, \; j \geq 0 \}$ as follows:
	\begin{equation}\label{defthetaj2}
	\theta_j^{[2]}=\sum_{i=0}^j\theta_i\theta_{j-i}.
	\end{equation}
}
We can now state our result on large gap asymptotics in full generality.
\begin{theorem}\label{thm:largegap}
Let $n\in\mathbb N$, $\tau_1,\ldots, \tau_{n-1}\in\mathbb R$, and let $F(s;1)=F(s;1;\tau_1,\ldots, \tau_{n-1})$ be the Fredholm determinant defined in \eqref{def:Fredholm} with $K$ given by \eqref{def:kernel}--\eqref{def p}.
As $s\to -\infty$, there exists a constant $C>0$, possibly depending on $n$ and the parameters $\tau_1,\ldots,\tau_{n-1}$, such that 
we have the {asymptotics}
\begin{equation}\label{FIntro}
 \log F(s;1) = -\sum_{\substack{j=0 \\ j \neq n+1}}^{2n}\frac{n^2}{(n+1-j)(2n+1-j)}\theta_j^{[2]}|s|^{\frac{2n-j+1}{n}}   + c \log |s| + \log C + o(1),
\end{equation}
{with $c=-\frac{1}{8}$ if $n=1$ and $c=-\frac{1}{2}$ otherwise.}
Moreover, the asymptotics \eqref{eq:qas-} can be improved to 
\begin{equation}
q({(-1)^{n+1}}s;1)=\sum_{i = 0}^{2n} \theta_i |s|^{\frac{1}{2n} - \frac{i}n} + \frac{c}{2\theta_0}|s|^{-2 - \frac{1}{2n}}+\mathcal O\left(|s|^{-2 - \frac{1}{n}}\right) ,\qquad \mbox{ as $s\to -\infty$.}\label{eq:qas-2}
\end{equation}
\end{theorem}
For $n = 1$, there are no parameters of deformation $\tau_i$, so that Theorem \ref{thm:largegap} gives the well known large gap asymptotics for the Tracy--Widom distribution \cite{BBdF, DIK,TracyWidomLevel}
$$
	\log F(s;1) = -\frac{|s|^3}{12} -\frac{1}8 \log |s| +  \log C + o(1) \quad \text{as} \; s\rightarrow -\infty. 
$$ 
In the first non-trivial case $n = 2$, we obtain
$$
	\log F(s;1) =  -\frac{2}{45}\sqrt{6}|s|^{5/2} - \frac{1}{12}\tau_1|s|^2 - \frac{\sqrt{6}}{54}\tau_1^2|s|^{3/2} - \frac{\sqrt{6}}{432}\tau_1^4|s|^{1/2} - \frac{1}2\log|s| +  \log C + o(1) \quad \text{as} \; s\rightarrow -\infty. 
$$
In the case $n = 3$, we have two deformation parameters $\tau_1,\tau_2$, and the large gap asymptotics read
\begin{align*}
	\log F(s;1)= &-{\frac {9}{560}}{20}^{\frac{2}{3}}{|s|}^{7/3}-\frac{1}{20}\tau_{{2}}{|s|}^{2}+{\frac {3\sqrt [3]{20}}{1000}} \left( 10\tau_{{1}}-3\tau_{{2}}^{2}
 \right) {|s|}^{5/3}+\frac{3}{2000} {20}^{\frac{2}{3}}\tau_{{2}} \left( 5\tau_{{1}}-\tau_{{2}}^{2} \right) {|s|}^{4/3} \\
 &-\frac{\sqrt [3]{20}}{5000} \tau_2 \left( 50 \tau_1^{2}-25\tau_2^2\tau_1+3\tau_2^{4} \right) |s|^{2/3}+\frac{{20}^{\frac{2}{3}}}{900000}\left(1000\tau_1^{3}-1800\tau_2^2\tau_1^{2}+630\tau_{2}^{4}\tau_{1}-63 \tau_{{2}}^{6} \right) |s|^{1/3}\\
			&-\frac{1}2 \log |s| +  \log C +  o(1) \quad \text{as} \; s\rightarrow -\infty.
\end{align*}
When $n$ increases, formulas get longer and longer, but they remain always completely explicit.\\

Let us end this introduction with an outline for the rest of this paper.
In Section \ref{sec:2}, we will show that the Fredholm determinants $F(s;\varrho)$ are equal to Fredholm determinants of a simpler integral operator, obtained through Fourier conjugation. This crucial observation will allow us to show that $F(s;\varrho)$ can be written as the Fredholm determinant of an operator which is of integrable type. We then use the formalism developed by Its, Izergin, Korepin, and Slavnov \cite{IIKS} to express the logarithmic derivative of $F$ in terms of a Riemann--Hilbert (RH) problem. We show that this RH problem is equivalent to a special case of the RH problem associated to the Painlev\'e II hierarchy. This enables us to prove formula \eqref{eq:TWidentity}.
In Section \ref{sec:3}, we will apply the Deift--Zhou steepest descent method to obtain $s\to +\infty$ asymptotics for the Painlev\'e II RH problem and to prove the asymptotics \eqref{eq:qas+0} for the Painlev\'e transcendent $q\big((-1)^{n+1}s;\varrho\big)$. 
In Section \ref{sec:4}, we perform a similar but somewhat more complicated asymptotic analysis as $s\to -\infty$, which will lead to the proof of \eqref{eq:qas-}, and to the proof of the large gap asymptotics stated in Theorems \ref{thm:largegaptau0} and \ref{thm:largegap}.
An important ingredient in this analysis is the construction of a suitable $g$-function, for which we need to exploit results related to the theory of topological minimal models of type $A_n$. The general strategy to prove Theorem \ref{thm:largegaptau0} and Theorem \ref{thm:largegap} shows similarities with the one from \cite{ClaeysGirottiStivigny}, where large gap asymptotics were obtained for Fredholm determinants arising from product random matrices.
Finally, in Appendix \ref{appendix:genAiry}, we prove that the kernels \eqref{def:kernel} define point processes.

\section{Painlev\'e expression for the Fredholm determinant}\label{sec:2}
This section is devoted to the proof of \eqref{eqthm1} and \eqref{eq:TWidentity}. In the following we denote with $\chi_R$ the indicator function of the contour $\gamma_R$, and analogously for $\gamma_L$. Let $K$ be a kernel of the form \eqref{def:kernel}. We will start by showing that the Fredholm determinant $\det(\Id -  \left. \varrho\, \mathcal{K}\right|_{[s,+\infty)})$ is equal to the Fredholm determinant of another integral operator $\mathcal L_s$ with kernel $L_s$ which is of integrable type, according to the terminology of Its, Izergin, Korepin, and Slavnov \cite{IIKS}. This means that $L_s(\lambda,\mu)$ has the form 
\be\label{IIKSform}
		L_s(\lambda, \mu) = \frac{f^{\mathrm{T}}(\lambda)g(\mu)}{\lambda - \mu},
	\ee
for some column vectors $f,g$ which are such that $f^{\mathrm{T}}(\lambda)g(\lambda)=0$.

\begin{proposition}\label{mainthm}
	Let $L_s$ be the kernel of the form \eqref{IIKSform} with
	\be
		f(\lambda) := \frac{1}{2 \pi i} \left( \begin{array}{c}
								{\rm e}^{\frac{(-1)^{n+1}}{2}p_{2n+1}(\lambda)} \chi_R(\lambda)\\
								\\
								{\rm e}^{\frac{(-1)^{n}}{2}p_{2n+1}(\lambda) + s \lambda} \chi_L(\lambda)
				\end{array} \right), \quad g(\mu) := \left( \begin{array}{c}
								{\rm e}^{\frac{(-1)^{n}}{2}p_{2n+1}(\mu)} \chi_L(\mu)\\
								\\
								{\rm e}^{\frac{(-1)^{n+1}}{2}p_{2n+1}(\mu) - s \mu} \chi_R(\mu)
				\end{array} \right),
	\ee
	{and $\mathcal L_s$ the associated operator acting on $L^2(\gamma_L\cup\gamma_R)$.}
Then we have 
	\be\label{eqdets}
		\mathrm{det}(\Id - \, \varrho\, \mathcal{L}_s) = \det \le(\Id -  \left. \varrho\,\mathcal{K}\right|_{[s,+\infty)} \ri).
	\ee
\end{proposition}
\proof

Using the fact that $L^2(\gamma_L \cup \gamma_R) = L^2(\gamma_L) \oplus L^2(\gamma_R)$, one can write the operator $\Id -  \varrho\, \mathcal L_s$ in block form as
\be
	\Id -\varrho \, \mathcal{L}_s = \left( \begin{array}{c|c}
				\Id & -\sqrt\varrho\, \F \\
				\hline
				-\sqrt\varrho\, \G & \Id
			\end{array} \right)
\ee
where $\F : L^2(\gamma_L) \longrightarrow L^2(\gamma_R)$ and $\G : L^2(\gamma_R) \longrightarrow L^2(\gamma_L)$ are respectively given by the kernels
\be
	F(\mu,\lambda) = \frac{1}{2 \pi i} \frac{{\rm e}^{\frac{(-1)^{n+1}}{2}(p_{2n+1}(\mu)-p_{2n+1}(\lambda))}}{\mu - \lambda}, \quad \quad G(\xi,\mu) = \frac{1}{2 \pi i} \frac{{\rm e}^{\frac{(-1)^{n}}{2}(p_{2n+1}(\xi)-p_{2n+1}(\mu))+ s\xi - s\mu}}{\xi - \mu}.
\ee
{This implies that the operators} $\F$ and $\G$ (and hence $\mathcal{L}_s$) are Hilbert--Schmidt, since
\begin{gather}
\|\F \|_2^2 = \frac{1}{4 \pi^2} \int_{\gamma_L} |\d \lambda| \int_{\gamma_R} | \d\mu | \frac{{\rm e}^{\mathrm{Re}\left[p_{2n+1}(\mu) - p_{2n+1}(\lambda)\right]}}{|\mu - \lambda| ^2} < + \infty, \label{FoperatorHS}
\end{gather}
and analogously for $\G$. We { need a bit more and} prove that actually $\F$ and $\G$ are also \emph{trace class}, {such} that the Fredholm determinant of $\Id - \varrho\, \mathcal L_s$ is well defined. To see this, introduce a third contour $\gamma_0 := i\mathbb{R} + \varepsilon$ which does not intersect with $\gamma_L$ and $\gamma_R$. We start with $\G$ and we introduce two operators
\be
	\mathcal C^{(1)} : L^2(\gamma_R) \longrightarrow L^2(\gamma_0) ; \quad \mathcal C^{(2)} : L^2(\gamma_0) \longrightarrow L^2(\gamma_L)
\ee
with kernels
\be
	C^{(1)}(\zeta,\mu) := \frac{{\rm e}^{\frac{(-1)^{n+1}}{2} p_{2n + 1}(\mu) - s\mu}}{2\pi i(\zeta - \mu)}; \quad C^{(2)}(\xi,\zeta) := \frac{{\rm e}^{\frac{(-1)^{n}}{2} p_{2n + 1}(\xi) + s\xi} }{2\pi i(\zeta - \xi)}.
\ee
With the same computation as in \eqref{FoperatorHS}, we prove that both $\mathcal C^{(1)}$ and $\mathcal C^{(2)}$ are Hilbert-Schimdt operators. Moreover, a residue argument shows that 
$$
	\G = \mathcal C^{(2)} \circ \mathcal C^{(1)},
$$  
hence $\G$ is a trace-class operator. For the operator $\F$, the same reasoning applies if we define the operators
\be
	\mathcal D^{(1)} : L^2(\gamma_L) \longrightarrow L^2(\gamma_0) ; \quad \mathcal D^{(2)} : L^2(\gamma_0) \longrightarrow L^2(\gamma_R)
\ee 
with kernels
\be
	D^{(1)}(\zeta,\lambda) := \frac{{\rm e}^{\frac{(-1)^{n}}{2} p_{2n + 1}(\lambda) }}{2\pi i(\lambda - \zeta)}; \quad D^{(2)}(\mu,\zeta) := \frac{{\rm e}^{\frac{(-1)^{n+1}}{2} p_{2n + 1}(\mu)} }{2\pi i(\zeta - \mu)}.
\ee

Hence we can use the following chain of equalities \bea\label{KsGF}
	\mathrm{det}(\Id - \varrho \, \mathcal{L}_s) &=& \mathrm{det}\left( \begin{array}{c|c}
				\Id & -\sqrt\varrho \,\mathcal{F} \\
				\hline
				-\sqrt\varrho\, \mathcal{G} & \Id
			\end{array} \right) = \mathrm{det} \Bigg[ \left( \begin{array}{c|c}
				\Id & -\sqrt\varrho \,\mathcal{F} \\
				\hline
				-\sqrt\varrho \,\mathcal{G} & \Id
			\end{array} \right) \left( \begin{array}{c|c}
				\Id & \sqrt\varrho\mathcal{F} \\
				\hline
				0 & \Id
			\end{array} \right) \Bigg] 
			\nonumber \\
			&=& \det(\Id - \varrho \,\G \circ \F)
\eea
where $\G \circ \F : L^2(\gamma_L) \longrightarrow L^2(\gamma_L)$ has kernel
{\be
	\frac{1}{(2 \pi i)^2} {\rm e}^{(-1)^n \frac{p_{2n+1}(\xi)+p_{2n+1}(\lambda) }{2} + s \xi} \int_{\gamma_R} \d \mu \frac{{\rm e}^{(-1)^{n+1}p_{2n+1}(\mu) - s\mu)}}{(\xi - \mu)(\mu - \lambda)}.
\ee 
}
{
It is now important to observe that the above arguments are valid also if we replace the contour $\gamma_L$ by $i\mathbb R$.
It is indeed easy to see that if we replace $\gamma_L$ by $i\R$ in \eqref{FoperatorHS}, the operators $\mathcal G$ and $\mathcal F$ are still Hilbert-Schmidt, therefore $\mathcal G \circ \mathcal F$ is still trace class.
 Moreover, replacing $\gamma_L$ by $i\mathbb R$ does not modify the Fredholm determinant $\det(\Id - \varrho \,\G \circ \F)$, because in the series definition of the Fredholm determinant the contour $\gamma_L$ can be deformed to $i\R$, by analyticity of the kernel. 
}

The last step of the proof {consists} in performing a conjugation of $\G \circ \F$ by the Fourier-type transform
\bea
	\mathfrak T :  L^2(i\R) &\longrightarrow &L^2(\R) \nonumber\\
		f(\xi) & \longmapsto & \frac{1}{\sqrt{2 \pi i}} \int_{i\mathbb R} f(\xi){\rm e}^{(-1)^{n+1}\frac{p_{2n+1}(\xi)}{2} - \xi x} \d\xi
\eea
with inverse
\bea
	\mathfrak T^{-1} :  L^2(\R) &\longrightarrow &L^2(i\R) \nonumber\\
		h(x) & \longmapsto & \frac{{\rm e}^{(-1)^n\frac{p_{2n+1}(\xi)}{2}}}{\sqrt{2 \pi i}} \int_{\R} h(x){\rm e}^{\xi x} \d x .
\eea
The resulting operator $\mathfrak T \circ \G \circ \F \circ \mathfrak T^{-1}$ has the following kernel, with $x,y\in\mathbb R$,
{\begin{multline}\label{GFAi}
	\frac{1}{(2 \pi i)^2} \int_{i\mathbb R} \frac{\d \xi}{2 \pi i}{\rm e}^{\xi(s - x)} \int_{\gamma_R} \d \mu \int_{i\mathbb R} \d \lambda \frac{{\rm e}^{(-1)^{n+1}(p_{2n+1}(\mu)-p_{2n+1}(\lambda))- s\mu}}{(\xi - \mu)(\mu - \lambda)}{\rm e}^{y \lambda} \nonumber  \\
								\\
								= \left\{\begin{array}{cc}
											-\displaystyle \frac{1}{(2 \pi i)^2} \int_{\gamma_R} \d \mu \int_{\gamma_L} \d \lambda \frac{{\rm e}^{(-1)^{n+1}(p_{2n+1}(\mu)-p_{2n+1}(\lambda)) -x \mu + y \lambda}}{\mu - \lambda} \quad & \mathrm{if} \quad x \geq s \nonumber\\
											\\
											0  &  \mathrm{if} \quad x < s
									\end{array}   \right.
\end{multline}}
{where we passed from the first to the second line by deforming the outer integration contour $i\mathbb R$ to the right/left depending the sign of $(x - s)$, and by computing the residue in $\xi$.} Combining the equations \eqref{KsGF} and the one above we {obtain}
\be
		\mathrm{det} (\Id - \varrho \, \mathcal{L}_s) =  \det(\Id - \varrho \, \G \circ \F) = \det (\Id - \varrho \, \mathfrak T \circ \G \circ \F \circ \mathfrak T^{-1}) = \det\le(\Id - \left. \varrho\,\mathcal{K}\right|_{[s,+\infty)}\ri) \nonumber
\ee
proving \eqref{eqdets}. \qed

Now, following \cite{IIKS}, we associate to the \emph{integrable} kernel $L_s(\lambda,\mu)$ a RH problem {with} jumps $$J(\zeta) := I - 2 \pi i \varrho\,  f(\zeta)g^{\mathrm{T}}(\zeta)$$ on the contour $\gamma_L \cup \gamma_R$ and regular asymptotics {at infinity.  It {was proved} in \cite[Appendix A]{BertolaCafasso1} that the logarithmic derivative of $F(s;\varrho)$, if it exists, can be expressed in terms of the solution $\Gamma$ of a RH problem. 
{Note that for $\varrho\in(0,1]$, we have that $F(s;\varrho)$ is smooth and positive, such that $\frac{\d}{\d s} \log F(s;\varrho)$ exists for all $s\in\mathbb R$. Then, we have the identity}}\footnote{{In the more general identity stated in \cite[Appendix A]{BertolaCafasso1}, an additional (explicit) term $H(J)$ is present, but this term vanishes in our situation.}}
\be\label{Marcoformula}
	\frac{\d}{\d s} \log F(s;\varrho) = \int_{\gamma_R \cup \gamma_L} \mathrm{Tr} \Bigg[ \Gamma_-^{-1} (\zeta)\Gamma'_-(\zeta) (\partial_s J)(\zeta) J^{-1}(\zeta) \Bigg] \frac{\d \zeta}{2 \pi i},
\ee
where $\Gamma'$ is the derivative of $\Gamma$ with respect to $\zeta$, and $\Gamma$ is the unique solution to the RH problem {below.} 

\subsubsection*{RH problem for $\Gamma$}
\begin{itemize}
\item[(a)] {$\Gamma:\mathbb C\setminus(\gamma_L\cup\gamma_R)\to \mathbb C^{2\times 2}$ is analytic,}
\item[(b)] {$\Gamma$ has continuous boundary values $\Gamma_\pm$ as $\zeta\in \gamma_L\cup\gamma_R$ is approached from the left ($+$) or right ($-$) side, and they are related by}
	\begin{align}	
				\Gamma_+(\zeta) &=  \Gamma_-(\zeta) \begin{pmatrix}1 & 0 \\ -\sqrt\varrho\,  {\rm e}^{-(-1)^{n+1}p_{2n+1}(\zeta)+s\zeta} & 1 \end{pmatrix}, &&\mbox{ for $\zeta\in\gamma_L$},\\ 
				\Gamma_+(\zeta) &= \Gamma_-(\zeta) \begin{pmatrix}1 & -\sqrt\varrho\,{\rm e}^{(-1)^{n+1}p_{2n+1}(\zeta)-s\zeta} \\ 0 & 1 \end{pmatrix}, &&\mbox{ for $\zeta\in\gamma_R$},
				\end{align} 
\item[(c)] there exists a matrix $\Gamma_1$ independent of $\zeta$ (but depending on $n, \varrho,\tau_j$ and $s$) such that $\Gamma$ satisfies
\be
				\Gamma(\zeta) = I + \Gamma_1\zeta^{-1}+\mathcal O(\zeta^{-2}), \quad \zeta \rightarrow \infty. 
		  \label{RHPGamma}
	\ee
\end{itemize}

The identity \eqref{Marcoformula} yields the following result.
{\begin{proposition}\label{corlogF}
	The Fredholm determinant $F(s;\varrho)$ satisfies the differential identity 
	\be\label{eqlogF}
		\frac{\d}{\d s} \log F(s;\varrho) = \Gamma_{1,11},
	\ee
	where $\Gamma_1$ is defined by \eqref{RHPGamma}.
\end{proposition}
}
	\proof 
	Starting from formula \eqref{Marcoformula}, 
	 we compute the integral on the right hand side using residues, as done in \cite{BertolaIsoTau}. We briefly repeat the procedure (which simplifies in our case) here. We start by noting that the jump matrix can be factorized as
	$$J(\zeta) = {\rm e}^{T_s(\zeta)} J_0 {\rm e}^{-T_s(\zeta)}, \quad T_s(\zeta) :=\left(\frac{(-1)^{n+1}}{2}p_{2n+1}(\zeta)-\frac{1}{2}s\zeta\right)\sigma_3,$$
	where $J_0$ is a (piecewise) constant matrix on the contours $\gamma_L, \gamma_R$. Here we are using the Pauli notations
$$
	\sigma_1 := \begin{pmatrix} 0 & 1\\ 1 & 0 \end{pmatrix}, \; \sigma_2 := \begin{pmatrix} 0 & -i\\ i & 0 \end{pmatrix}, \; \sigma_3 := \begin{pmatrix} 1 & 0\\ 0 & -1 \end{pmatrix}.
$$

	  Hence, we obtain
	  \begin{multline*}
	  	  \displaystyle \int_{\gamma_R \cup \gamma_L} \mathrm{Tr} \Bigg[ \Gamma_-^{-1}(\zeta) \Gamma'_-(\zeta) (\partial_s J)(\zeta) J^{-1}(\zeta) \Bigg] \frac{\d \zeta}{2 \pi i} = 
		 \displaystyle \int_{\gamma_R \cup \gamma_L} \mathrm{Tr} \Bigg[ \Gamma_-^{-1}(\zeta) \Gamma'_-(\zeta) \Big(\partial_s T(\zeta) - J(\zeta) \partial_s T(\zeta) J^{-1}(\zeta) \Big) \Bigg] \frac{\d \zeta}{2 \pi i}  \nonumber\\
= \int_{\gamma_R \cup \gamma_L} \mathrm{Tr} \Bigg[ \Gamma_-^{-1}(\zeta) \Gamma'_-(\zeta)(\partial_s T)(\zeta)\Bigg] - \mathrm{Tr} \Bigg[\Gamma_+^{-1}(\zeta) \Gamma'_+(\zeta)(\partial_s T)(\zeta)\Bigg] \frac{\d \zeta}{2 \pi i}, \nonumber
	  \end{multline*}
	where, in the last line, we used the {cyclic property of the trace and the form of $J$}. Now, we observe that the integral contributions of the minus sides of $\gamma_L$ and $\gamma_R$ sum up to zero, and that the plus sides can be deformed to a {large counterclockwise circle $C_R$ of radius $R$}, leading to 
{	\be
		\frac{\d}{\d s} \log F(s;\varrho) = \lim_{R\to +\infty}\oint_{C_R} \mathrm{Tr}\le[ \Gamma^{-1}(\zeta) \Gamma'(\zeta)\frac{1}{2}\zeta\sigma_3\ri]  \frac{\d\zeta}{2 \pi i} = \frac{1}2\mathrm{Tr}\Big[\Gamma_{1,11} -\Gamma_{1,22}\Big] = \Gamma_{1,11}.
	\ee
	In} the last equality we used the fact that $\Gamma_1$ { has zero trace}, since $\Gamma$ has determinant equal to $1$.
	  \QED

The relation with the second Painlev\'e hierarchy is established {by transforming the RH problem for $\Gamma$ to one} with {constant} jumps. Namely, define
\begin{equation}
\label{def:Psi}
\Psi(\zeta) := {\sigma_3}\Gamma(2 i\zeta){\sigma_3}{\rm e}^{T_s(2i\zeta)}.
\end{equation}

Then it is {straightforward} to verify that $\Psi$ solves the following RH problem:

\subsubsection*{RH problem for $\Psi$}

\begin{itemize}
\item[(a)] {$\Psi:\mathbb C\setminus (\gamma_U\cup \gamma_D)\to \mathbb C^{2\times 2}$ is analytic,  with $\gamma_U, \gamma_D$ as in \figurename \ \ref{Fig2},}
\item[(b)]
the boundary values of {$\Psi$} on $\gamma_U \cup \gamma_D$
are related by the conditions
    \begin{align}
        \label{RHP Psi2: b1}
        \Psi_+(\zeta) &= \Psi_-(\zeta)
            \begin{pmatrix}
                1 & 0 \\
                \sqrt\varrho & 1
            \end{pmatrix},
            && \mbox{for $\zeta\in \gamma_U$,} 
        \\      
        \label{RHP Psi2: b2}
            \Psi_+(\zeta) &= \Psi_-(\zeta)
            \begin{pmatrix}
                1 & -\sqrt\varrho \\
                0 & 1
            \end{pmatrix},
            && \mbox{for $\zeta\in  \gamma_D$,}
    \end{align}
\item[(c)] $\Psi$ has the following behavior at infinity:
    \begin{equation}\label{RHP Psi2: c}
        \Psi(\zeta)=\le[ I+\frac{\Psi_1}{\zeta}+\mathcal O(\zeta^{-2})\ri] {\rm e}^{- i\Theta(\zeta;s,\tau_1,\ldots, \tau_{n-1})\sigma_3},
        \qquad \mbox{as $\zeta\to\infty$,} 
        \end{equation}
        where 
        \begin{equation}
        \Psi_1=-\frac{i}{2} {\sigma_3}\Gamma_1{\sigma_3}
        \label{normalization}
        \end{equation}
and
\be\label{def Theta}
        \Theta(\zeta;s,\tau_1, \ldots,\tau_{n-1})=\frac{1}{4n+2}(2\zeta)^{2n+1}+
\sum_{j=1}^{n-1}\frac{(-1)^{{n+j}}
        \tau_{j}}{4j+2}(2\zeta)^{2j+1}+s\zeta.
\ee
\end{itemize}
Note that we have now chosen to orient both jump contours from left to right, so that $\gamma_D=-i\gamma_R$, but with reverse orientation, and $\gamma_U = - i\gamma_L$.
{Note also that the RH solution $\Psi$ depends on the choice of jump contours $\gamma_U, \gamma_D$, but that the solution corresponding to a different choice of jump contours can be obtained by analytic continuation. The matrix $\Psi_1$ is independent of the choice of jump contours.}

The above RH problem {is equivalent to the one} in \cite[Section 4.2]{CIK}, in the case where $\alpha=0$ and where the Stokes multipliers are given by
\[s_1=-s_{2n+1}=\sqrt\varrho,\qquad s_2=\ldots=s_{2n}=0.\]

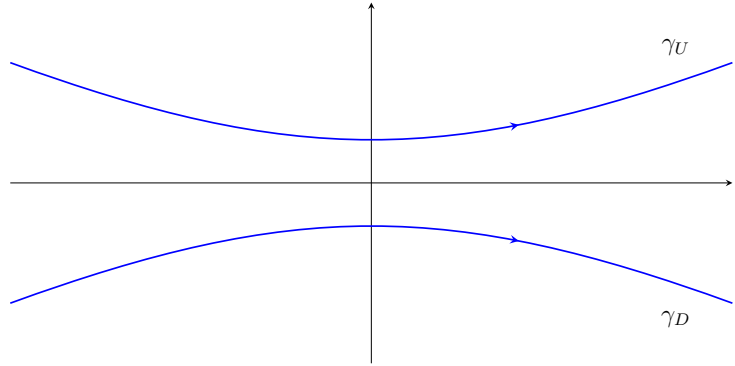
\begin{figure}[h!]
\centering
\scalebox{.8}{
\begin{tikzpicture}[>=stealth]
\path (0,0) coordinate (O);

\draw[->] (-6,0) -- (6,0) coordinate (x axis);
\draw[->] (0,-3) -- (0,3) coordinate (y axis);


\draw[thick, blue, ->- = .7]  (-6,2) .. controls + (-20:5cm) and + (200:5cm) .. (6,2);
\node [above] at (5,2) {{ \large $ \gamma_U$}};

\draw[thick, blue, ->- = .7]  (-6,-2) .. controls + (20:5cm) and + (-200:5cm) .. (6,-2);
\node [below] at (5,-2) {{ \large $ \gamma_D$}};
\end{tikzpicture}}
\caption{The contours $\gamma_U$ and $\gamma_D$. \label{Fig2}}
\end{figure}

\begin{proposition}
	Let $n\in\mathbb N$ and let $\Psi$ be the solution to the above RH problem. Then, the function 
\begin{equation}\label{qdef}
		q(s;\varrho) := \begin{cases}
			2i\Psi_{1,12}(s)= -2i\Psi_{1,21}(s), &\text{for} \; n \; \text{odd,} \\
			\\
			2i\Psi_{1,12}(-s)= -2i\Psi_{1,21}(-s), &\text{for} \; n \; \text{even,} 
	\end{cases}
\end{equation}
{is real for $s\in\mathbb R$, $0<\rho\leq 1$, and it}
solves the Painlev\'e II hierarchy equation \eqref{PIIn}.
\end{proposition}
\proof
{
We start by showing that $\Psi(\zeta)$ satisfies two symmetry relations, provided that we choose $\gamma_U$ and $\gamma_D$ in such a way that\footnote{To be more precise, $\gamma_U=\overline{\gamma}_D$ as oriented contours, while $\gamma_U = -\gamma_D$ as sets, but they have opposite orientation.}
$\gamma_U=\overline{\gamma}_D=-\gamma_D$ . First, we verify that $\sigma_1\Psi(-\zeta)\sigma_1$ satisfies the RH problem for $\Psi$. But since the RH solution $\Psi$ is unique (this follows from standard arguments using Liouville's theorem), we obtain the identity
\[\sigma_1\Psi(-\zeta)\sigma_1=\Psi(\zeta).\]
Similarly we obtain, for real values of $s$, the relation
\[\sigma_1\overline{\Psi(\overline{\zeta})}\sigma_1=\Psi(\zeta).\]
These identities imply in particular that
\[-\Psi_1=\sigma_1\Psi_1\sigma_1=\overline{\Psi}_1.\]
It follows that the two expressions at the right hand side of \eqref{qdef} are indeed equal, and moreover that $q$ is real-valued.

The remaining part of the proof} consists of identifying $\Psi$ with solutions to the Lax pair associated with the Painlev\'e II hierarchy (cf. \cite{CJM,FN,Kudryashov_2001,MazzoccoMo}). Here we use the same approach and normalization as in \cite[Section 4.1]{CIK}. Namely, if we define
{\[ 
Z(\zeta;s) := \begin{cases}
				\Psi(-i\zeta;s), &\text{for} \; n \; \text{odd,} \\
				\\
				\Psi(-i\zeta;-s), &\text{for} \; n \; \text{even,} 
			\end{cases}
\]
}
then one verifies using the RH conditions that $Z$ satisfies the differential equation
\begin{equation}\label{LaxEquation1}
	\dfrac{\partial}{\partial s}Z(\zeta;s) = \begin{pmatrix} -i \zeta & q(s;\varrho) \\ q(s;\varrho) & i \zeta  \end{pmatrix} Z(\zeta;s)
\end{equation}
{with $q$ as in \eqref{qdef}.} In addition, we have a second differential equation of the form
\begin{equation}\label{LaxEquation2}
	\dfrac{\partial}{\partial \zeta}Z(\zeta;s)  = M(\zeta;s) Z(\zeta;s),
\end{equation}
with $M$ as in \cite[Equations (4.1)--(4.3)]{CIK}. {The compatibility between the Lax equations \eqref{LaxEquation1} and \eqref{LaxEquation2} now implies, in the same way as in \cite{CIK}, that $q=q(s;\varrho)$ solves the Painlev\'e II hierarchy equation \eqref{PIIn}.} \QED

{
\begin{remark}
In particular, the solution $q(s;\varrho)$ is uniquely determined by the Stokes parameters described in the Riemann--Hilbert for $\Psi$.
\end{remark}
}

We are now ready to prove \eqref{eqthm1}.

\begin{proposition}
Let $F$ be defined by \eqref{def:Fredholm}, and $q$ by \eqref{qdef}. Then, we have the differential identity
	\be\label{doublelogderivative}
		\frac{\d ^2}{\d s^2} \log F(s;\varrho) = - q^2\big({(-1)^{n+1}s};\varrho\big).
	\ee
\end{proposition}
\proof
{It follows from \eqref{LaxEquation1} that $\le(\partial_s \Psi \ri) \Psi^{-1}$ is a polynomial in $\zeta$. Hence, the term in $\zeta^{-1}$ in the large $\zeta$ expansion of $ \le(\partial_s \Psi \ri) \Psi^{-1}$ vanishes, and this yields
\begin{gather}
	-i[\Psi_2,\sigma_3] + i[\Psi_1,\sigma_3]\Psi_1 + \partial_s\Psi_1 = 0, \label{Lax}
\end{gather}
where $\Psi_1$ and $\Psi_2$ are defined by the expansion} 
$$ \Psi(\zeta)=\le[ I+\frac{\Psi_1}{\zeta}+ \frac{\Psi_2}{\zeta^2}+\mathcal O\le(\frac{1}{\zeta^{3}}\ri)\ri] {\rm e}^{- i\Theta(\zeta)\sigma_3} \; \text{as} \; \zeta \to \infty .$$
In particular, the $(1,1)$-entry {of \eqref{Lax}} gives 
\be\label{Psieq}
	\frac{\d}{\d s} \Psi_{1,11} = -2i \le(\Psi_{1,12}\ri)^2.
\ee
Using \eqref{Psieq} together with {Proposition \ref{corlogF}} and the relation
$
\Psi_{1,11} = -\frac{i}{2} \Gamma_{1,11}$ (see \eqref{normalization}), we obtain the result. 
\QED
{
Since $q$ is real-valued, the above proposition implies that $\log F(s;\varrho)$ is concave as a function of $s$. We also know that $F$ is a distribution function for $0< \rho\leq 1$, so we can conclude that
$$	\lim_{s \rightarrow + \infty} \log F(s;\varrho) = 0, \quad \lim_{s \rightarrow +\infty} (\log F)'(s;\varrho) = 0.
$$}
Hence, integrating \eqref{doublelogderivative}, we obtain
\be\label{Ffirstder}
\frac{\d}{\d s} \log F(s;\varrho) =\int_{s}^{+\infty}q^2\big((-1)^{n+1}\xi;\varrho\big)\d\xi,
\ee
and after another integration by parts we obtain \eqref{eq:TWidentity}.

\section{Asymptotics for the Painlev\'e function $q\left((-1)^{n+1}s;\varrho\right)$ as  $s\to + \infty$}\label{sec:3}

We start from the RH problem for $\Psi$, see \eqref{RHP Psi2: b1}--\eqref{RHP Psi2: c}, and we rescale it in the following way. Let
\begin{gather}
\Xi(\zeta) := \Psi\le(s^{\frac{1}{2 n}} \zeta \ri){\rm e}^{i\Theta\left(s^{\frac{1}{2 n}} \zeta;s,\tau_1,\ldots, \tau_{n-1}\right)\sigma_3},
\end{gather}
then $\Xi$ satisfies the following RH problem.

\subsubsection*{RH problem for $\Xi$}
\begin{itemize}
\item[(a)] $\Xi$ is analytic in $\mathbb C\setminus\left(\tilde\gamma_U\cup\tilde\gamma_D\right)$, where {$\tilde \gamma_U=s^{-\frac{1}{2n}}\gamma_U$ and $\tilde \gamma_D=s^{-\frac{1}{2n}}\gamma_D$} are the transformed contours under the scaling $\zeta \mapsto s^{\frac{1}{2 n}} \zeta$,
\item[(b)] $\Xi$ satisfies the jump relations
 \begin{align}
        &\Xi_+(\zeta) = \Xi_-(\zeta)
            \begin{pmatrix}
                1 & 0 \\
                \sqrt\varrho\,{\rm e}^{2i\Theta\left(s^{\frac{1}{2 n}} \zeta;s,\tau_1,\ldots, \tau_{n-1}\right) } & 1
            \end{pmatrix},
            && \mbox{for $\zeta\in \tilde \gamma_U$,} 
        \\      
           & \Xi_+(\zeta) = \Xi_-(\zeta)
            \begin{pmatrix}
                1 & -\sqrt\varrho\,{\rm e}^{-2i\Theta\left(s^{\frac{1}{2 n}} \zeta;s,\tau_1,\ldots, \tau_{n-1}\right)} \\
                0 & 1
            \end{pmatrix},
            && \mbox{for $\zeta\in \tilde \gamma_D$,}
            \end{align}
            \item[(c)]
            as $\zeta\to\infty$, we have
\[\Xi(\zeta)= I+ \frac{\Xi_1}{\zeta}+\mathcal O(\zeta^{-2}),\]
where
\begin{gather}
\Xi_1 = s^{-\frac{1}{2 n}}\Psi_1. \label{Xi1Psi1}
\end{gather}
\end{itemize}

We now note that, {by \eqref{def Theta},}
\be \label{expansionTheta}
\Theta\left(s^{\frac{1}{2 n}} \zeta;s,\tau_1,\ldots, \tau_{n-1}\right)=s^{\frac{2n+1}{2n}}\le[\frac{(2\zeta)^{2 n + 1}}{4 n + 2} + \zeta \ri]+\mathcal O\left(s^{\frac{2n-1}{2n}}\zeta^{2n-1}\right)
\ee
as $s\to +\infty$, uniformly for $\zeta\in\tilde \gamma_U\cup \tilde \gamma_D$ and for $\tau_1,\ldots, \tau_{n-1}$ in any compact set. For this reason we
want to {choose} the contours  $\tilde \gamma_U$ and $\tilde \gamma_D$ in such a way that 
\begin{align}
&\operatorname{Im} \le[\frac{(2\zeta)^{2 n + 1}}{4 n + 2} + \zeta \ri] > 0,  \qquad {\zeta} \in \tilde \gamma_U,\\
&\operatorname{Im} \le[\frac{(2\zeta)^{2 n + 1}}{4 n + 2} +\zeta \ri] < 0,  \qquad {\zeta} \in \tilde \gamma_D.
\end{align}
Even in the case where the parameters {$\tau_1,\ldots, \tau_{n-1}$} are not all equal to zero, the behavior of the phase $\Theta$ is  still asymptotically dominated by the leading terms $\frac{2\zeta^{2 n + 1}}{4 n + 2} + \zeta$ as $s \to +\infty$.
It is {straightforward to check that we can always choose $\tilde \gamma_U$ and $\tilde \gamma_D$ such that the inequalities above are satisfied. This is also illustrated in \figurename \ \ref{fig2}.}

\begin{figure}
\centering
\includegraphics[width=.5\textwidth]{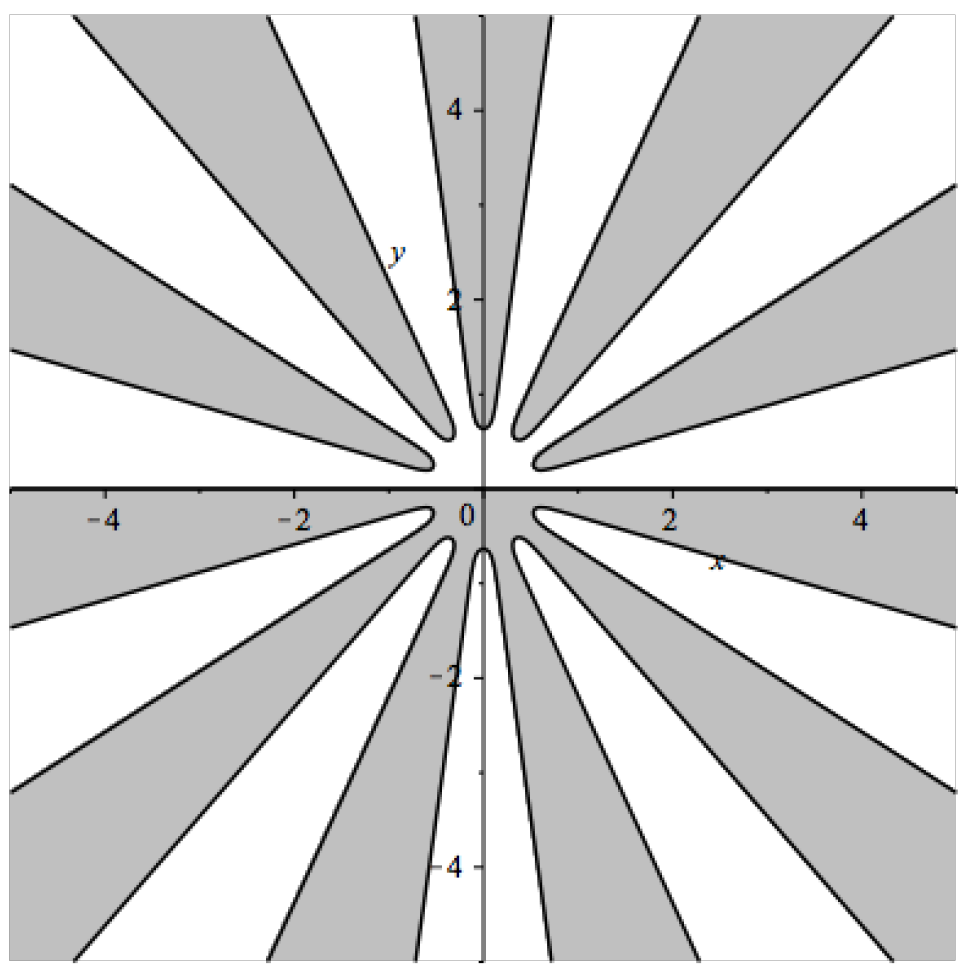}
\caption{The plot of the  imaginary part of the quantity $\frac{(2\zeta)^{2n + 1}}{4n + 2} + \zeta$ in the case $n = 5$. The white areas are the zones where $\operatorname{Im} \le[\frac{(2\zeta)^{2n + 1}}{4n + 2} +\zeta \ri] > 0$ while the grey areas are the zones where $\operatorname{Im} \le[\frac{(2\zeta)^{2n + 1}}{4n + 2} +\zeta \ri] <0$.  \label{fig2}}
\end{figure}

{
With such a choice of $\tilde\gamma_U$ and $\tilde\gamma_D$, the jump relation for $\Xi$ reads as $s\to +\infty$,
\[\Xi_+(\zeta) = \Xi_-(\zeta)\left(I+\mathcal O(\zeta^{-2}{\rm e}^{-C s^{\frac{2 n + 1}{2 n}}})\right),\] 
for a suitable constant $C>0$, uniformly in $\zeta\in\tilde\gamma_U\cup\tilde\gamma_L$.
By small norm theory for RH problems (see e.g.\ \cite{DeiftCourant, smallnormRH}),} we can infer that there exists a constant $C>0$ such that
\begin{gather}
\label{eq:Xias}\Xi(\zeta) = I + \mathcal{O}\le(\zeta^{-1}{\rm e}^{-C s^{\frac{2 n + 1}{2 n}}}\ri) \qquad \text{as } s\to +\infty,
\end{gather}
uniformly in $\zeta\in\mathbb C\setminus\left(\tilde \gamma_U\cup \tilde \gamma_D\right)$. This implies in particular, by  \eqref{qdef} and \eqref{Xi1Psi1}, that
\[ q\big((-1)^{n+1}s;{\varrho}\big) = 2i \Psi_{1,12}(s)= 2is^{\frac{1}{2n}} \Xi_{1,12}(s)=\mathcal{O}\le({\rm e}^{-C s^{\frac{2 n + 1}{2 n}}}\ri),\]
as $s\to +\infty$, which proves \eqref{eq:qas+0}.

If we want more precise asymptotics, we need to use the integral equation
\[
\Xi(\zeta)=I + \sqrt\varrho \int_{\tilde \gamma_U} \Xi_-(w)\frac{  {\rm e}^{2i\Theta\left(s^{\frac{1}{2 n}} w;s,\tau_1,\ldots, \tau_{n-1}\right) } }{w-\zeta} \frac{\d w}{2\pi i}  \sigma_- -\sqrt\varrho  \int_{\tilde \gamma_D}\Xi_-(w)\frac{  {\rm e}^{-2i\Theta\left(s^{\frac{1}{2 n}} w;s,\tau_1,\ldots, \tau_{n-1}\right) } }{w-\zeta} \frac{\d w}{2\pi i} \sigma_+,
\]
which is satisfied by $\Xi$ as a consequence of the RH conditions,
with $\sigma_+ = \begin{pmatrix} 0& 1\\0&0\end{pmatrix}$ and $\sigma_- = \begin{pmatrix} 0& 0\\1&0\end{pmatrix}$. To see this, one verifies that the right hand side of the above equation satisfies the three RH conditions for $\Theta$. Uniqueness of the RH solution then yields the integral equation (see, for instance, \cite{FIKN}).
This integral equation { implies that
\[
\Xi_1(s) = -\sqrt\varrho \int_{\tilde \gamma_U}  \Xi_-(w) {\rm e}^{2i\Theta\left(s^{\frac{1}{2 n}} w;s,\tau_1,\ldots, \tau_{n-1}\right) } \frac{\d w}{2\pi i}  \sigma_- +\sqrt\varrho \int_{\tilde \gamma_D}  \Xi_-(w){\rm e}^{-2i\Theta\left(s^{\frac{1}{2 n}} w;s,\tau_1,\ldots, \tau_{n-1}\right) } \frac{\d w}{2\pi i} \sigma_+.\]
Using also \eqref{eq:Xias}, we obtain
\[
\Xi_{1,12}(s) = \sqrt\varrho \int_{\tilde \gamma_D}  {\rm e}^{-2i\Theta\left(s^{\frac{1}{2 n}} w;s,\tau_1,\ldots, \tau_{n-1}\right) } \frac{\d w}{2\pi i} (1+o(1))\]
as $s\to +\infty$.
It follows from \eqref{qdef} that
\[ q\left((-1)^{n+1}s; \varrho \right)=2i\Psi_{1,12}(s)=2is^{\frac{1}{2n}}\Xi_{1,12}(s)=\frac{1+o(1)}{\pi}\sqrt\varrho\,  s^{\frac{1}{2n}}\int_{\tilde \gamma_D}   {\rm e}^{-2i\Theta\left(s^{\frac{1}{2 n}} w;s,\tau_1,\ldots, \tau_{n-1}\right) }  \d w\]
as $s\to +\infty$.}
If $\tau_1=\ldots=\tau_{n-1}=0$, it follows directly from the substitution $u=2is^{\frac{1}{2n}}w$ and \eqref{def:genAiry} that 
\begin{equation}\label{eq:asqAiry}
q\left((-1)^{n+1}s;\varrho \right) =  \sqrt\varrho\, \Ai_{2n+1}(s) (1+o(1)) ,\qquad \mbox{as $s\to +\infty$,}
\end{equation} 
as already announced in the introduction.

\begin{remark}\label{remark2as}
In order to prove \eqref{eq:asqAiry} for general $\tau_1,\ldots, \tau_{n-1}$, we would need the asymptotics of the integral $\ds\int_{\tilde \gamma_D}\!\!  {\rm e}^{-2i\Theta\left(s^{\frac{1}{2 n}} w;s,\tau_1,\ldots, \tau_{n-1}\right) } \d w$. These can in principle be obtained using saddle point arguments, but this would require a lot of technical work, especially because the integrand depends on the large parameter $s$ in a rather complicated manner.
We therefore do not aim to prove this.
\end{remark}

\section{Asymptotics for $q\left((-1)^{n+1}s ;1\right)$ as $s\to -\infty$}\label{sec:4} 

In this section, we restrict to the case where $\varrho=1$, and we analyze the RH problem for $\Psi$ in the limit $s\to -\infty$, in order to derive asymptotics for $q\left((-1)^{n+1}s;1 \right)$ as $s\to -\infty$ and large gap asymptotics for $F(s;1)$.

\subsection{Construction of the $g$-function}
In the asymptotic analysis as  $s\to +\infty$, we were able to normalize the RH problem for $\Psi$ simply by multiplying at the right by the factor ${\rm e}^{i\Theta\left(s^{\frac{2n+1}{2n}}\zeta;s,\tau_1,\ldots, \tau_{n-1}\right)\sigma_3}$. However, for $s$ negative, the topology of the set $\left\{\zeta: \Im \Theta (s^{\frac{2n+1}{2n}}\zeta;s,\tau_1,\ldots, \tau_{n-1} )>0\right\}$ is different from the one sketched in Figure \ref{fig2}, and forces us to normalize the RH problem for $\Psi$ in a different way, namely by means of the construction of a suitable $g$-function.

\medskip

We define a function $g$ of the form
\begin{equation}\label{def g-}
g(\zeta)=\sum_{j=1}^nc_{j}(\zeta^2-\zeta_0^2)^{\frac{2j+1}{2}},
\end{equation}
with $(\zeta^2-\zeta_0^2)^{\frac{2j+1}{2}}$ analytic in  $\mathbb C\setminus [-\zeta_0,\zeta_0]$ and such that it behaves like $\zeta^{2j+1}$ as $\zeta\to\infty$.

We fix the constants $c_j$ and the branch point $\zeta_0>0$ by
the requirement that there exists a value $g_1(s)$ such that
\begin{equation}\label{gtheta-}
|s|^{\frac{2n+1}{2n}}g(\zeta)=\Theta(|s|^{\frac{1}{2n}}\zeta)+\frac{g_1(s)}{\zeta}+\mathcal O(\zeta^{-2}),\qquad\mbox{ as $\zeta\to\infty$.}
\end{equation}
It was shown in \cite[Section 4.6.2]{CIK} that this yields the following asymptotic formulas as $s\to -\infty$,
\begin{align}
&\label{def cn}c_n=\frac{2^{2n}}{2n+1},\\
&\label{def cj}c_{n-m}=\frac{\Gamma(n+\frac{1}{2})}{\Gamma(n-m+\frac{3}{2})}\frac{2^{2n-1}}{m!}\zeta_0^{2m}
+\mathcal O(|s|^{-\frac{1}{n}}), & m=1, \ldots , n-1,
\end{align}
and
\be
\label{zeta0}
\zeta_0^{2n}=\frac{n!\sqrt{\pi}}{2^{2n}\Gamma(n+\frac{1}{2})}+
\mathcal O(|s|^{-\frac{1}{n}})=\frac{n!^2}{(2n)!}+\mathcal O(|s|^{-\frac{1}{n}}).
\ee

The aim of the rest of this subsection is to give more explicit asymptotics for $\zeta_0$ {and $g_1$}, which will be needed later on. 
\begin{lemma}\label{lemmazeta0}
Suppose that $|s|$ is sufficiently large, set $\tau_n=1$, and let $\zeta_0=\zeta_0(s)$ be the unique positive solution to the equation
	\begin{equation}\label{shortzetaequat}
		\sum_{k = 1}^n  (-1)^{n-k} {2k \choose k}\tau_k |s|^{\frac{k-n}{n}} \zeta_0^{2k}= 1.
	\end{equation}
If we define 
\begin{equation}\label{cvalues}
c_{n-m}  = 
\sum_{k=0}^{m}  (-1)^{m-k}2^{2(n-m+k)-1} \tau_{n-m+k} |s|^{-\frac{m-k}{n}} \frac{\Gamma\le(n-m+k +\frac{1}{2}\ri)}{k! \Gamma\le(n-m+\frac{3}{2}\ri)} \zeta_0^{2k},  
\end{equation}
for $m=0,\ldots, n-1$,
then $g$ defined by \eqref{def g-} has the asymptotics \eqref{gtheta-}.
\end{lemma} 
\proof
If we expand $g(\zeta)$ defined by \eqref{def g-} and $\Theta(|s|^\frac{1}{2n}\zeta)$ defined by \eqref{def Theta} as $\zeta\to\infty$, we obtain in a straightforward manner that \eqref{gtheta-} implies 
\begin{multline*}
 |s|^{\frac{2n+1}{2n}} \sum_{j=1}^n c_{j} \le[\sum_{k=0}^j   \frac{(-1)^k\Gamma\le( j + \frac{3}{2}\ri)}{k! \Gamma\le( j-k+\frac{3}{2}\ri)} \zeta_0^{2k} \zeta^{2j+1-2k}\ri] =
 \Bigg( \sum_{j=1}^{n}(-1)^{n+j}\frac{2^{2j} \tau_{j}}{2j+1}\zeta^{2j+1} |s|^{\frac{2j+1}{2n}} \Bigg)\\
 - |s|^{\frac{2n+1}{2n}}\zeta +\mathcal O(1),
\end{multline*}
as $\zeta\to\infty$. Equating the coefficients of $\zeta^{2j + 1}, j =1,\ldots,n,$ of the left and right hand side, we obtain the system of equations
\begin{align}
&  \sum_{\ell= 0}^{n-j} c_{j+\ell}   \frac{(-1)^\ell\Gamma\le( j+\ell + \frac{3}{2}\ri)}{\ell ! \Gamma\le( j+\frac{3}{2}\ri)} \zeta_0^{2\ell}  =  (-1)^{n+j}\frac{2^{2j} \tau_{j}}{2j+1} |s|^{-1 + \frac{j}{n}} \qquad j=1,\ldots, n \label{cequations}\\
 &  \sum_{j=1}^n c_{j}   \frac{(-1)^j  \Gamma\le( j + \frac{3}{2}\ri)}{j! \Gamma\le(\frac{3}{2}\ri)} \zeta_0^{2j}   = -1. \label{2.5}
\end{align}
The $n$ equations in \eqref{cequations} form a linear system with unknowns $c_1,\ldots,c_n$ and $\zeta_0$ as a parameter in triangular form, whose solution is given by \eqref{cvalues}.
Substituting this into \eqref{2.5}, we obtain
 \begin{flalign}\label{longzetaequat}
 -1   
& = \sum_{m=0}^{n-1} \le[  \sum_{k=0}^{m}  (-1)^{m-k}\frac{2^{2(n-m+k)} \tau_{n-m+k}}{2(n-m+k)+1} |s|^{-\frac{m-k}{n}}  \frac{\Gamma\le(n-m+k +\frac{3}{2}\ri)}{k!}  \frac{(-1)^{n-m} }{(n-m)! \Gamma\le(\frac{3}{2}\ri)} \zeta_0^{2(n-m+k)}\ri]\nonumber \\
 &= \sum_{m=0}^{n-1} \le[  \sum_{j=0}^{m} (-1)^{j}\frac{2^{2(n-j)} \tau_{n-j}}{2(n-j)+1} |s|^{-\frac{j}{n}}  \frac{\Gamma\le(n-j +\frac{3}{2}\ri)}{(m-j)!}  \frac{(-1)^{n-m} }{(n-m)! \Gamma\le(\frac{3}{2}\ri)} \zeta_0^{2(n-j)}\ri]\nonumber \\
  &=\sum_{j=0}^{n-1}  (-1)^j \tau_{n-j} |s|^{-\frac{j}{n}}  \frac{2^{2(n-j)}}{2(n-j) + 1}\frac{\Gamma\le(n-j +\frac{3}{2}\ri)}{\Gamma\le(\frac{3}{2}\ri)} 
  \le[  \sum_{m=j}^{n-1}     \frac{(-1)^{n-m} }{(m-j)!(n-m)! }\ri]  \zeta_0^{2(n-j)}. 
 \end{flalign} 
 Moreover,
 \[
\frac{2^{2(n-j)}}{2(n-j)+1}   \frac{\Gamma\le(n-j +\frac{3}{2}\ri)}{\Gamma\le(\frac{3}{2}\ri)} 
=  2^{2(n-j)}   \frac{\Gamma\le(n-j +\frac{1}{2}\ri)}{\Gamma\le(\frac{1}{2}\ri)}
= 
2^{2(n-j)}  \frac{\le(2(n-j)-1\ri)!!}{2^{n-j}}  =
\frac{(2(n-j))!} {(n-j)!}
\] 
and
 \[
 \sum_{k=0}^{n-1-j}     \frac{(-1)^{n-j-k} }{k!(n-j-k)! } =  -\frac{1}{(n-j)!},
\]
 by the binomial formula.  Substituting the latter two identities into \eqref{longzetaequat}, we obtain \eqref{shortzetaequat}. For sufficiently large $|s|$, we observe easily that \eqref{shortzetaequat} has a unique positive solution, since it reduces to 
 \[{2n \choose n}\zeta_0^{2n}=1\] in the large $|s|$ limit. \QED

{For $g_1$}, as a first step, we need the following identity.
\begin{lemma}\label{lemma:g1}Set $\tau_n=1$. The constant $g_1=g_1(s)$ defined by \eqref{gtheta-} is given by
	\begin{equation}\label{g1equat}
		g_1(s) = \frac{1}2\sum_{k = 1}^n (-1)^{n-k}\tau_k{2k \choose k-1}|s|^{\frac{2k+1}{2n}}\zeta_0(s)^{2k + 2}.
	\end{equation}
\end{lemma}
\proof
The proof is analogous to that of Lemma \ref{lemmazeta0}. Namely we {expand as $\zeta\to\infty$, but now we take into account one more term:}
\begin{multline*}
 |s|^{\frac{2n+1}{2n}} \sum_{j=1}^n c_{j} \le[\sum_{k=0}^{j+1}   \frac{(-1)^k\Gamma\le( j + \frac{3}{2}\ri)}{k! \Gamma\le( j-k+\frac{3}{2}\ri)} \zeta_0(s)^{2k} \zeta^{2j+1-2k}\ri] \\=
 \Bigg( \sum_{j=1}^{n}(-1)^{n+j}\frac{2^{2j} \tau_{j}}{2j+1}\zeta^{2j+1} |s|^{\frac{2j+1}{2n}} \Bigg)- |s|^{\frac{2n+1}{2n}}\zeta + \frac{g_1(s)}{\zeta}+\mathcal O(\zeta^{-2}),
\end{multline*}
as $\zeta\to\infty$,
and this leads to
$$
g_1(s) = |s|^{\frac{2n+1}{2n}}\sum_{j = 1}^n (-1)^j c_j \frac{ \Gamma(j + \frac{3}2)}{\Gamma(j+2)\Gamma(\frac{1}{2})}\zeta_0(s)^{2j + 2}.
$$
Plugging the values \eqref{cvalues} of $c_k,\, k=1,\ldots,n$ into the equation above and performing some straightforward computations as in \eqref{longzetaequat} we finally obtain \eqref{g1equat}. \QED

We will now derive a large $|s|$ expansion of $\zeta_0(s)$. First we need the following lemma: 
\begin{lemma}\label{lemma:LG}

Let $m$ be a positive integer and $\alpha_0,\ldots, \alpha_{m-1}\in\mathbb R$. Define the polynomial $\mu:= \mu(z)$ as
	$$\mu(z) := z^{m} + \sum_{j = 0}^{m-1} \alpha_j z^j.$$
For large positive $\mu$, there is a unique positive solution $z$ of the equation $\mu(z)=\mu$. It admits a Puiseux series at $\mu=\infty$ of the form
\begin{equation}\label{Puiseux}
	z(\sigma) = \sigma + \sum_{j = 0}^\infty \varphi_j \sigma^{-j}, \quad \quad \sigma := \mu^{1/m},
\end{equation}
with
$$
\varphi_0 = -\dfrac{1}m \alpha_{m - 1} \qquad\mbox{and}\qquad
		\varphi_j = \dfrac{1}{j} \res {z = \infty} \le(\mu^{\frac{j}{m}}(z)\ri), \quad \forall \, i \geq 1,
$$
where the residue at infinity is equal to minus the term in $1/z$ of the expansion of $\mu^{\frac{i}{m}}(z)$ as $z \rightarrow \infty$, and where we take the branch of $\mu^{\frac{i}{m}}(z)$ which is positive for large $z>0$.
\end{lemma}
\proof
It is a standard fact that there is a unique positive solution $z(\sigma)$ to the polynomial equation $\mu(z)=\sigma^m$ for $\sigma>0$ sufficiently large, and that it has an expansion of the form \eqref{Puiseux}. All the coefficients $\{\varphi_i, \; i\geq 0\}$ can be computed recursively by plugging \eqref{Puiseux} in the equation
$$ 
	z^m(\sigma) + \sum_{j = 0}^{m-1} \alpha_jz^j(\sigma) = \sigma^m,
$$
and by solving term by term. In this way, we immediately find the value of $\varphi_0$. For the other terms, it is convenient to rewrite the residue using $\sigma$ as coordinate. Indeed, since because of \eqref{Puiseux} we have
$$\frac{\d z}{\d \sigma} = \Big(1 - \sum_{j \geq 1} j \varphi_j \sigma^{-j-1}\Big),$$
then 
$$
	\dfrac{1}{i} \res {z = \infty} \mu^{\frac{i}{m}}(z) = \dfrac{1}{i} \res {\sigma = \infty} \Bigg(\sigma^i \Big(1 - \sum_{j \geq 1} j \varphi_j \sigma^{-j - 1}\Big) \Bigg) = \varphi_i, \quad \forall i \geq 1,
$$ 
which proves the stated result. \QED
\begin{remark}For the interested reader, we mention a connection between the above result and flat coordinates for a special class of Frobenius manifolds, the ones related to topological minimal models of type $A_n$.\\
	Consider a polynomial $\mu(z)$ as in {Lemma \ref{lemma:LG}. By shifting the variable $z\mapsto z+c$, we can assume} without any loss of generality that $\alpha_{m-1} = 0$.  Define the complex manifold $M$
	$$M := \Bigg\{\mu(z)  = z^m + \sum_{j = 0}^{m-2} \alpha_j z^j , \quad \alpha_0,\ldots,\alpha_{m-2} \in \mathbb C \Bigg\}.$$
	At each point $\mu \in M$ we identify the tangent space $T_{\mu} M$ with the algebra
	$$A_\mu := \faktor{\mathbb{C}[z]}{\mu'(z)}.$$
	Moreover, we equip this algebra with the bilinear pairing {$\langle\cdot,\cdot\rangle_\mu$ defined by
	$$
		\langle f(z),g(z)\rangle_{\mu} := \res{z = \infty} \frac{f(z)g(z)}{\mu'(z)}.
	$$}
	The multiplication on $T_\mu M$ and the bilinear pairing satisfy a certain compatibility condition which, together with some additional structures that are not important here, turns $M$ into a Frobenius manifold; the one associated to the singularity of type $A_{m-1}$ (see for instance \cite{Dij,Boris2DTFT} and references therein). The coefficients {$\varphi_1, \ldots, \varphi_{m-1}$} in Lemma \ref{lemma:LG} (up to a rescaling and a permutation necessary to identify the first coordinate with the unit vector field) are the flat coordinates associated to the given Frobenius structure (see, for instance, Corollary 4.6 in \cite{Boris2DTFT}).  With a different approach, flat coordinates of $A_m$ Frobenius manifolds are also used in \cite{ACvM1}, in relation with the same type of generalized Airy kernels.\label{remark:Frobenius}
\end{remark}

\begin{proposition}\label{expansionzeta0}
Let
$$\lambda(z) := \sum_{k=1}^{n}(-1)^{n-k} {2k \choose k} \tau_{k} z^{2k},\qquad \tilde\lambda(z) := \sum_{k=1}^{n}(-1)^{n-k} {2k \choose k} \tau_{k} z^k$$ and
	\begin{equation}\theta_i(\tau_1,\ldots,\tau_n) \equiv \theta_i := \begin{cases}
															\ds{2n \choose n}^{-\frac{1}{2n}}, &i = 0, \\[2ex] 
															\dfrac{1}{2i - 1} \res {z = \infty} \lambda^{\frac{2i-1}{2n}}(z), &i \geq 1.
														\end{cases}\label{def:thetaj}\end{equation}					
	Then, the positive solution of \eqref{shortzetaequat} has the asymptotic expansion
	\begin{equation}\label{zetasolut}
		\zeta_0(s) \sim \sum_{i = 0}^{\infty} \theta_i |s|^{-\frac{i}{n}} ,\qquad\mbox{ as $s\to -\infty$}.
	\end{equation}
	Moreover, 
	\begin{equation}\label{zetasolut2}
	\zeta_0^2(s) \sim \sum_{i = 0}^{\infty} \theta_i^{[2]} |s|^{-\frac{i}{n}},\qquad\mbox{ as $s\to -\infty$},\qquad \mbox{with}\qquad
	\theta_i^{[2]} := \sum_{k = 0}^i \theta_{k}\theta_{i - k},
\end{equation}	
or equivalently
\begin{equation}
\theta_i^{[2]}(\tau_1,\ldots,\tau_n) \equiv \theta_i^{[2]} = \begin{cases} \ds\theta_0,&i=0,\\[2ex]
														\ds\frac{\ds\tau_{n - 1}}{4n-2},&i = 1,\\[2ex]

\ds\dfrac{1}{i - 1} \res {z = \infty} \tilde\lambda^{\frac{i-1}{n}}(z), &i \geq 2. 							\end{cases}\label{formula:thetaj2}\end{equation}
In particular,
\begin{equation}\label{nullcoefficients}
	\theta_{kn + 1}^{[2]} = 0 \quad \forall\, k \geq 1.
\end{equation}
\end{proposition}

\proof 
We apply Lemma \ref{lemma:LG} to the polynomial
\[\mu(z)=
\sum_{k=1}^{n}(-1)^{n-k} {2k \choose k} \theta_0^{2k}\tau_{k} z^{2k}=\lambda\left(\theta_0 z\right)
\]
which satisfies $\mu\left(\theta_0^{-1}|s|^{\frac{1}{2n}}\zeta_0(s)\right)=|s|=:\sigma^{2n}$.
The lemma then implies that
\[\theta_0^{-1}\sigma\zeta_0(s)=\sigma+\sum_{j=0}^\infty\varphi_j\sigma^{-j}\]
as $\sigma\to\infty$, or equivalently
\[\zeta_0(s)=\theta_0+\theta_0\sum_{j=0}^\infty\varphi_j\sigma^{-j-1}.\]
The result \eqref{def:thetaj} now follows by observing that \[\varphi_{2i-1}=\theta_0^{-1}\res {z = \infty} \lambda^{\frac{2i-1}{2n}}(z),\]
and that
\[\varphi_{2i}=\theta_0^{-1}\res {z = \infty} \lambda^{\frac{i}{n}}(z)=0.\]
This directly implies \eqref{zetasolut2}. To see 
\eqref{formula:thetaj2}, we apply Lemma \ref{lemma:LG} to the polynomial $\tilde\lambda$.
From \eqref{formula:thetaj2} and the definition of $\tilde\lambda$, we deduce \eqref{nullcoefficients}.
 \QED
{Given the expansion of $\zeta_0$, we can now also deduce the one for} $g_1(s)$, as stated in the next proposition.
\begin{proposition}
There exists a constant $\kappa$ such that
	\begin{equation}\label{g_1expansion}
		g_1(s) \sim \sum_{\substack{ i \geq 0 \\ i \neq n+1}} \frac{n}{2(n + 1 - i)} \theta_i^{[2]}|s|^{\frac{2n + 1 - 2i}{2n}} + \kappa |s|^{-\frac{1}{2n}} \quad \mathrm{as} \; s\rightarrow -\infty.
	\end{equation}
\end{proposition}
\proof
Using \eqref{g1equat} and the asymptotic expansion of $\zeta_0(s)$ for $s \rightarrow -\infty$, we observe that $g_1(s)$ has an expansion of the form
\begin{equation}\label{ansatzg_1}
	g_1(s) \sim \sum_{i \geq 0} g_1^{[i]} |s|^{\frac{2n + 1 - 2i}{2n}}.
\end{equation}
Let us now recall the definition of $\tilde\lambda$ as in Proposition \ref{expansionzeta0}, such that
$\widetilde\lambda(z(\sigma))=\sigma^{2n}$ with $\sigma^{2n}=|s|$ and $z(\sigma):=\sigma^2\zeta_0^2(s)$.
We also 
write \[Q(z):=\sum_{k = 1}^n (-1)^{n-k}\tau_k{2k \choose k-1}z^{k + 1},\]
such that $Q(z(\sigma))=2\sigma g_1(s)$ by Lemma \ref{lemma:g1}.
Using the identity ${2k \choose k}k = {2k \choose k-1}(k+1)$, we obtain the relation
\[\frac{\d}{\d\sigma}Q(z(\sigma))=z(\sigma)\frac{\d}{\d\sigma}\tilde\lambda(z(\sigma)),\]
and this results in 
\[\frac{\d}{\d\sigma}(2\sigma g_1(s))=\frac{\d}{\d\sigma}Q(z(\sigma))=z(\sigma)\frac{\d}{\d\sigma}\tilde\lambda(z(\sigma))=2n\sigma^{2n+1}\zeta_0^2(s).\]
Combining this identity with \eqref{ansatzg_1} and \eqref{zetasolut2} and integrating in $\sigma$, we obtain the result, with unknown integration constant $\kappa$.
\QED

\subsection{Normalization of the RH problem}
We now take the jump contours $\gamma_U$ and $\gamma_D$ of a special form.
It is convenient for us to take them such that
 $\widehat\gamma_U:=|s|^{-\frac{1}{2n}}\gamma_U$ and $\widehat\gamma_D:=|s|^{-\frac{1}{2n}}\gamma_D$ 
coincide with the real line on the interval $[-\zeta_0,\zeta_0]$, and such that they are independent of $s$.
Then we define
\be
S(\zeta)=\Psi\left(|s|^{\frac{1}{2n}}\zeta\right){\rm e}^{i|s|^{\frac{2n+1}{2n}}g(\zeta)\sigma_3}. \label{Sdef}
\ee
$S$ now satisfies the following RH problem. 
\subsubsection*{RH problem for $S$}
\begin{itemize}
\item[(a)]  $S$ is analytic in $\mathbb C\setminus\Sigma_S$, with \[\Sigma_S=[-\zeta_0,\zeta_0]\cup \Sigma_1\cup\Sigma_2\cup\Sigma_3\cup\Sigma_4\]
as in Figure \ref{Fig4}.
\item [(b)] $S$ has the jump relations 
\begin{align}
&S_+(\zeta)=S_-(\zeta)\begin{pmatrix}0&-1\\ 1&{\rm e}^{-2i|s|^{\frac{2n+1}{2n}} g_+(\zeta) }
\end{pmatrix},&&\zeta\in(-\zeta_0,\zeta_0),\\
&S_+(\zeta)=S_-(\zeta)\begin{pmatrix}1&0\\ {\rm e}^{2i|s|^{\frac{2n+1}{2n}}g\left(\zeta\right)}&1
\end{pmatrix},&&\zeta\in\Sigma_1\cup\Sigma_2,\\
&S_+(\zeta)=S_-(\zeta)\begin{pmatrix}1&-{\rm e}^{-2i|s|^{\frac{2n+1}{2n}}g\left(\zeta\right)}\\0&1
\end{pmatrix},&&\zeta\in\Sigma_3\cup\Sigma_4,
\end{align}
which follows from the fact that $g_+(\zeta) + g_-(\zeta) = 0$ for $\zeta \in (-\zeta_0,\zeta_0)$.
\item[(c)] $S(\zeta)=I+S_1\zeta^{-1}+\mathcal O(\zeta^{-2})$ as $\zeta\to\infty$, 
with 
\begin{equation}\label{S1Psi1}S_1=|s|^{-\frac{1}{2n}}\Psi_1+ig_1(s)\sigma_3.
\end{equation}
\end{itemize}

\begin{figure}[h!]
\centering
\scalebox{.8}{
\begin{tikzpicture}[>=stealth]
\path (0,0) coordinate (O);

\draw[->] (-6,0) -- (6,0) coordinate (x axis);
\draw[->] (0,-3) -- (0,3) coordinate (y axis);

\draw[thick, blue, ->- = .7] (-2.5,0) -- (2.5,0);
\draw[fill, blue] (2.5,0) circle [radius=0.05];
\node[above] at (2.5,0.1) {${\color{blue} \large \zeta_0}$};
\draw[fill, blue] (-2.5,0) circle [radius=0.05];
\node[above] at (-2.5,0.1) {${\color{blue}\large -\zeta_0}$};

\draw[thick, blue, ->- = .7] (2.5,0) -- (6,2);
\draw[thick, blue, ->- = .7] (-6,2) -- (-2.5,0);

\draw[thick, blue, ->- = .7] (2.5,0) -- (6,-2);
\draw[thick, blue, ->- = .7] (-6,-2) -- (-2.5,0);

\node [above] at (-5,2) {{\color{blue} \large $ \Sigma_1$}};
\node [above] at (5,2) {{\color{blue} \large $ \Sigma_2$}};

\node [below] at (-5,-2) {{\color{blue} \large $ \Sigma_3$}};
\node [below] at (5,-2) {{\color{blue} \large $ \Sigma_4$}};

\end{tikzpicture}}
\caption{The jump contours $\Sigma_S$ for $S$. \label{Fig4}}
\end{figure}
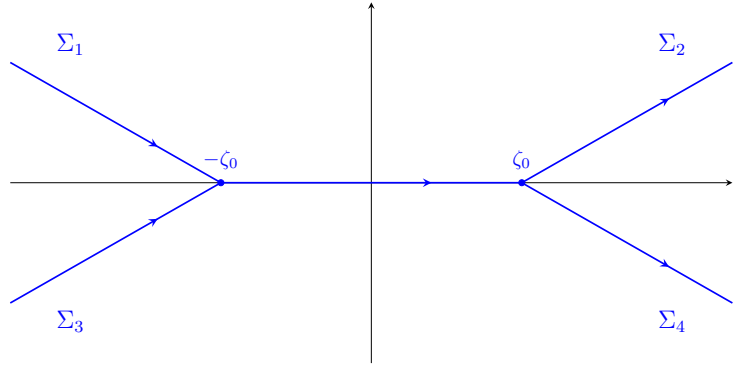

It is easy to see, by studying the argument of the terms in (\ref{def g-}), that the $g$-function satisfies the following relations:
\begin{align}
& \Im\le[ g(\zeta) \ri] >0 \qquad \zeta \in \Sigma_1\cup \Sigma_2, \label{ineq1}\\
&\Im\le[ g(\zeta) \ri] <0 \qquad \zeta \in \Sigma_3\cup \Sigma_4. \label{ineq2}
\end{align}
for sufficiently large $|s|${, if we choose the angles between $\Sigma_1,\ldots, \Sigma_4$ and the real line sufficiently small}. Furthermore, 
\begin{lemma}
\begin{gather}
\Im \le[ g_+(\zeta) \ri] <0 \qquad \zeta \in (-\zeta_0, \zeta_0) \label{ineq3}
\end{gather}
for sufficiently large $|s|$.
\end{lemma}
\begin{proof}
Following the arguments of  \cite[Section 4.6]{CIK}, we  consider the subinterval $[0,\zeta_0)$. For  $\zeta \in (-\zeta_0,0)$, the analysis is similar.  

We set
\begin{gather}
|\zeta^2 - \zeta^2_0| = \zeta_0^2 z, \qquad z \in [0,1).
\end{gather}
From the asymptotic behavior of the coefficients $c_{j}$ and the endpoint $\zeta_0 $ (see (\ref{def cn})--(\ref{zeta0})), we have that
\begin{gather}
g_+(\zeta) - g_-(\zeta) = 2g_+(\zeta) =  i n!\sqrt{\pi z}\zeta_0  \sum_{j=1}^n \frac{(-1)^j z^{j + \frac{1}{2}}}{\Gamma\le(j+\frac{3}{2}\ri) (n-j)!} +\mathcal O \le(|s|^{-\frac{1}{n}}\ri), \quad s\to -\infty,
\end{gather}
for $z \in [0,1)$ (i.e. $\zeta \in [0,\zeta_0)$). In order to prove that $\Im \le[ g_+(\zeta) \ri] <0$ on $[0,\zeta_0)$, we need to prove that the sum appearing in the above expression is negative. 

Using Jacobi polynomials $P^{(\alpha,\beta)}_n(x)$ and their representation as hypergeometric functions (in particular formula (22.5.42) from \cite{AbramowitzStegun}), we have that
\begin{gather}
P^{\le(\frac{3}{2}, -n-\frac{1}{2}\ri)}_{n-1}(1-2z) = - \frac{\Gamma\le(n+\frac{3}{2}\ri)}{z} \le( \sum_{j=1}^n \frac{(-1)^j z^{j + \frac{1}{2}}}{\Gamma\le(j+\frac{3}{2}\ri) (n-j)!}\ri)
\end{gather}
On the other hand, using an alternative hypergeometric representation for Jacobi polynomials (formula (22.5.45) from \cite{AbramowitzStegun}), we have
\begin{gather}
P^{\le(\frac{3}{2}, -n-\frac{1}{2}\ri)}_{n-1}(1-2z) = \frac{2\Gamma\le(n+\frac{3}{2}\ri)}{\sqrt{\pi}} \sum_{m=0}^{n-1} \frac{z^{n-m-1}(1-z)^m}{m!(n-m-1)!\le(n-m+\frac{1}{2}\ri)!} >0
\end{gather}
for $z\in [0,1)$, and this implies the result.
\end{proof}

\subsection{Construction of the global parametrix}

From \eqref{ineq1}, \eqref{ineq2}, and \eqref{ineq3}, it follows that the jump matrices for $S$ are exponentially close to the identity matrix as $s\to -\infty$, as long as we stay away from the interval $[-\zeta_0,\zeta_0]$.
For this reason, we can expect that, away from $\pm\zeta_0$, a first approximation to $S$ will be given by a function which has the same jump on $(-\zeta_0,\zeta_0)$ and the same asymptotic behavior as $S$. This motivates the RH conditions for the global parametrix given below.

\subsubsection*{RH problem for $P^\infty$}

\begin{itemize}
\item[(a)]  $P^\infty$ is analytic in $\mathbb C\setminus[-\zeta_0,\zeta_0]$,
\item[(b)] $P^\infty$ has the jump relation
\begin{gather}
P^\infty_+(\zeta) = P^\infty_-(\zeta) \begin{pmatrix} 0 & -1\\ 1& 0\end{pmatrix}  , \qquad \zeta \in (-\zeta_0,\zeta_0),
\end{gather}
\item[(c)] $P^\infty(\zeta)=I+P^\infty_1\zeta^{-1}+\mathcal O(\zeta^{-2})$ as $\zeta\to\infty$.
\end{itemize}
If we impose in addition that $P^\infty(\zeta)=\mathcal O\le((\zeta\mp\zeta_0)^{-1/4}\ri)$ as $\zeta\to \pm\zeta_0$, then it is a well-known fact{, see e.g.\ \cite{DeiftCourant},} that the solution $P^\infty$ is given by
\begin{equation}\label{def:Pinf}
P^\infty(\zeta)=\frac{1}{2}
\begin{pmatrix}1&i\\i&1\end{pmatrix}\gamma(\zeta)^{\sigma_3}
\begin{pmatrix}1&-i\\-i&1\end{pmatrix},\qquad \gamma(\zeta)=\left(\frac{\zeta+\zeta_0}{\zeta-\zeta_0}\right)^{1/4},
\end{equation}
where $\gamma(\zeta)$ is analytic on $\mathbb C\setminus[-\zeta_0, \zeta_0]$ and $\gamma(\zeta)\to 1$ as $\zeta\to\infty$. 

Indeed,
the jump relation for $P^\infty$ can be easily verified by noticing that $\gamma_+(\zeta) = -i\gamma_-(\zeta)$ for $\zeta \in (-\zeta_0,\zeta_0)$. The asymptotic behavior at infinity follows by expanding the function $\gamma$: more precisely, we have 
\be\label{Pinf:as}
P^\infty(\zeta)=I+\frac{\zeta_0}{2\zeta}\begin{pmatrix}0&-i\\i&0\end{pmatrix}+\mathcal O(\zeta^{-2}) \qquad \text{as } \zeta\to\infty.
\ee

\subsection{Construction of local Airy parametrices}	
As $P^\infty$ blows up as $\zeta\to \pm\zeta_0$, we cannot expect it to be a good approximation to $S$ close to $\pm\zeta_0$.	Let us therefore take sufficiently small disks $U_\pm$ around $\pm\zeta_0$, and construct local parametrices $P$ in these disks, which have the same jump relations as $S$ on the parts of the jump contour inside $U_\pm$, and which match with the global parametrix $P^\infty$ on $\partial U_\pm$. To build the parametrices, we use a standard construction using the Airy function, see e.g.\ \cite{DeiftCourant, smallnormRH} {as general reference, and \cite{ClaeysGrava}
for a construction which is almost identical to the present one.} 

Define the functions $y_j$ in terms of the Airy function,
\[
    y_j=y_j(\lambda )=\omega^j\Ai(\omega^j\lambda ),\qquad j=0,1,2,
\]
with $\omega=e^{\frac{2\pi i}{3}}$, and define the functions
\begin{align*}
    &\Phi_{1}(\lambda ) := \sqrt{2\pi}
        \begin{pmatrix}
            y_0 & {-}y_2\\
            {-}iy_0' & iy_2'
        \end{pmatrix}, 
        \\
        &\Phi_2(\lambda ) := \sqrt{2\pi}\begin{pmatrix}
            -y_1 & {-}y_2\\
            iy_1' & iy_2'
        \end{pmatrix},         \\
        &\Phi_3(\lambda ) := \sqrt{2\pi}\begin{pmatrix}
            -y_1 & y_0\\
            iy_1' & -iy_0'
        \end{pmatrix}.
\end{align*}
The functions $\Phi_1$, $\Phi_2$, $\Phi_3$ are entire functions, and by the Airy function identity $y_0+y_1+y_2=0$, they satisfy the relations
{\begin{align}
&\Phi_{1}(\lambda ) = \Phi_{3}(\lambda )
  \begin{pmatrix} 0 & {-}1 \\ 1&1 \end{pmatrix},\\ 
&\Phi_1(\lambda )=\Phi_2(\lambda )\begin{pmatrix} 1 & 0 \\ 1&1\end{pmatrix} ,\\
&\Phi_2(\lambda )=\Phi_3(\lambda )\begin{pmatrix} 1 & {-}1 \\ 0&1\end{pmatrix}.
\end{align}
}
Moreover, if we define sectors  $S_1 := \{\lambda :0<\arg \lambda <\pi -\delta\}$, $S_2 := \{\lambda :\frac{\pi}{3}+\delta< \arg \lambda  < \frac{5\pi}{3}-\delta\}$, and $S_3 := \{\lambda :\pi+\delta<\arg \lambda  < 2\pi\}$, then we have
\begin{equation}\label{Psiasympt}
\Phi_j (\lambda ) = \frac{1}{\sqrt{2}}
\lambda ^{-\frac{\sigma_3}{4}} \begin{pmatrix} 1& i\\ i & 1\end{pmatrix} \le[ I + \frac{1}{48\lambda ^{3/2}}\begin{pmatrix}1&6i\\6i&-1\end{pmatrix}+{\mathcal O (\lambda ^{-3})}\ri] {{\rm e}^{\mp\frac{2}{3} \lambda^{\frac{3}{2}}\sigma_3}}
\end{equation}
{as  $\lambda \to\infty$ in the sector $S_j$, where we take the plus sign in the lower half plane and the minus sign in the upper half plane.
}

Let us now focus on the disk $U_-$. We define a conformal map $\lambda(\zeta)$ by the equation
\begin{gather}
\lambda(\zeta)^{\frac{3}{2}} = \frac{3}{2}{\rm e}^{\mp \frac{3\pi i}{2}}g(\zeta),
\end{gather}
where the minus sign is taken in the upper half plane and the plus sign in the lower half plane.
Then it is straightforward to verify that $\lambda$ maps $U_-$ to a neighbourhood of $\lambda =0$. 

We now define the local parametrix as
\begin{equation}\label{def:P}
P(\zeta) = A(\zeta) \Phi_j
\le( |s|^{\frac{2n+1}{3n}} \lambda(\zeta) \ri) {\rm e}^{i |s|^{\frac{2n+1}{2n}} g(\zeta)\sigma_3},
\end{equation}
where we take $j$ as follows: $j=1$ if {$\zeta$ is in the sector between $(-\zeta_0,+\infty)$ and $\Sigma_1$, $j=2$ for $\zeta$ in the sector between $\Sigma_1$ and $\Sigma_3$, and $j=3$ for $\zeta$ in the sector between $\Sigma_3$ and $(-\zeta_0,+\infty)$.}
If $A$ is analytic in $U_-$, then $P$ satisfies the same jump relations as $S$ inside $U_-$, and we additionally fix the matrix $A(\zeta)$ by the requirement that
\begin{gather}
P(\zeta) P^\infty(\zeta)^{-1} = I + \mathcal{O}\le( |s|^{-\frac{2n+1}{2n}}\ri)
\end{gather}
uniformly for $\zeta\in\partial U_-$,
as $s\to -\infty$. Therefore, we take
\begin{gather}
A(\zeta)  
= P^\infty(\zeta)\frac{1}{\sqrt{2}}\begin{pmatrix} 1&-i\\ - i&1 \end{pmatrix} \le(|s|^{\frac{2n+1}{3n}}\lambda(\zeta)\ri)^{\frac{\sigma_3}{4}}.
\end{gather}

We can make the matching condition more precise and have, using \eqref{def:Pinf} and \eqref{Psiasympt},
\begin{multline}\label{matching2}
P(\zeta) P^\infty(\zeta)^{-1} = I +  \frac{1}{96|s|^{\frac{2n+1}{2n}}\lambda(\zeta)^{3/2}}{\begin{pmatrix}7 \gamma(\zeta)^{2}-5\gamma(\zeta)^{-2}&7i \gamma(\zeta)^{2}+5i\gamma(\zeta)^{-2}\\
7i \gamma(\zeta)^{2}+5i\gamma(\zeta)^{-2}&-7 \gamma(\zeta)^{2}+5\gamma(\zeta)^{-2}
\end{pmatrix}}\\+\mathcal{O}\le( |s|^{- \frac{2n+1}{n}}\ri)
\end{multline}
uniformly for $\zeta\in\partial U_-$,
as $s\to -\infty$.

The parametrix near $+\zeta_0$ can be constructed similarly. Namely, we set 
$P(\zeta)=\sigma_1P(-\zeta)\sigma_1$ for $\zeta\in U_+$.
We then have
\begin{multline}\label{matching3}
P(\zeta) P^\infty(\zeta)^{-1} = I +  \frac{1}{96|s|^{\frac{2n+1}{2n}}\lambda(-\zeta)^{3/2}} { \begin{pmatrix}5\gamma(\zeta)^{2} -7 \gamma(\zeta)^{-2}&5i\gamma(\zeta)^{2}+7i \gamma(\zeta)^{-2}\\
5i\gamma(\zeta)^{2}+7i \gamma(\zeta)^{-2}&-5\gamma(\zeta)^{2} +7 \gamma(\zeta)^{-2} \end{pmatrix}} \\+\mathcal{O}\le( |s|^{- \frac{2n+1}{n}}\ri)
\end{multline}
uniformly for $\zeta\in\partial U_+$,
as $s\to -\infty$.

\subsection{Small norm RH problem}
Finally, we define $R$ as
\begin{equation}\label{defR-}
R(\zeta)=\begin{cases}
S(\zeta)P^\infty(\zeta)^{-1},&\zeta\in\mathbb C\setminus(\overline{U_+\cup U_-})\\
S(\zeta)P(\zeta)^{-1},&\zeta\in U_\pm.
\end{cases}
\end{equation}
where $P$ {stands for} both local parametrices at the endpoints $\pm \zeta_0$.

The function $R$ solves the RH problem:
\subsubsection*{RH problem for $R$}
\begin{itemize}
\item[(a)]  $R$ is analytic in $\mathbb C\setminus\Sigma_R$, where $\Sigma_R$ is depicted in Figure \ref{Fig5},
\item [(b)] $R$ has the jump relations
\begin{gather}
R_+(\zeta)=R_-(\zeta) J_R(\zeta),\qquad\zeta\in\Sigma_R,
\end{gather}
where $J_R$ takes the form
{\begin{equation}
J_R(\zeta)=\begin{cases} I + \mathcal O \le({\rm e}^{-C |s|^{\frac{2n+1}{2n}}}/(\zeta^2+1) \ri),  &\zeta \in \Sigma_R \setminus (\partial U_+ \cup \partial U_- ) \\
I + J^{(1)}(\zeta)|s|^{-\frac{2n+1}{2n}}+J^{(2)}(\zeta)|s|^{-\frac{2n+1}{n}}+\mathcal{O}\le( |s|^{-  \frac{6n+3}{2n} }\ri), &\zeta \in \partial U_\pm,\end{cases}
\end{equation}
for some $C>0$, uniformly in $\zeta$ as $s\to -\infty$,}
with \begin{equation}
J^{(1)}(\zeta)=
\begin{cases}
\displaystyle  \frac{1}{96\,\lambda(-\zeta)^{3/2}} {\begin{pmatrix}5\gamma(\zeta)^{2} -7 \gamma(\zeta)^{-2}&5i\gamma(\zeta)^{2}+7i \gamma(\zeta)^{-2}\\
5i\gamma(\zeta)^{2}+7i \gamma(\zeta)^{-2}&-5\gamma(\zeta)^{2} +7 \gamma(\zeta)^{-2} \end{pmatrix}},&\mbox{on $\partial U_+$},\\
\displaystyle  \frac{1}{96\, \lambda(\zeta)^{3/2}}{\begin{pmatrix}7 \gamma(\zeta)^{2}-5\gamma(\zeta)^{-2}&7i \gamma(\zeta)^{2}+5i\gamma(\zeta)^{-2}\\
7i \gamma(\zeta)^{2}+5i\gamma(\zeta)^{-2}&-7 \gamma(\zeta)^{2}+5\gamma(\zeta)^{-2}
\end{pmatrix}},&\mbox{on $\partial U_-$,}
\end{cases}\end{equation}
by \eqref{matching2} and \eqref{matching3}, and for some matrix $J^{(2)}$ independent of $s$,
\item[(c)] $R(\zeta)=I+R_1\zeta^{-1}+\mathcal O(\zeta^{-2})$ as $\zeta\to\infty$.
\end{itemize}

	\begin{figure}[h!]
\centering
\scalebox{.8}{
\begin{tikzpicture}[>=stealth]
\path (0,0) coordinate (O);

\draw[thick, blue, ->- = .7] (2.5,0) -- (6,2);
\draw[thick, blue, ->- = .5] (-6,2) -- (-2.5,0);

\draw[thick, blue, ->- = .7] (2.5,0) -- (6,-2);
\draw[thick, blue, ->- = .5] (-6,-2) -- (-2.5,0);

\draw[fill, white] (-2.5,0) circle (.7cm);
\draw[->, blue, thick] (-3.2,0) arc (180:0:.7);
\draw[blue, thick] (-1.8,0) arc (0:-180:.7);
\draw[fill, blue] (-2.5,0) circle [radius=0.05];
\node[above] at (-2.5,0.1) {${\color{blue} \large - \zeta_0}$};

\draw[fill, white] (2.5,0) circle (.7cm);
\draw[->, blue, thick] (1.8,0) arc (180:0:.7);
\draw[blue, thick] (3.2,0) arc (0:-180:.7);
\draw[fill, blue] (2.5,0) circle [radius=0.05];
\node[above] at (2.5,0.1) {${\color{blue} \large \zeta_0}$};

\node [above] at (-5,2) {{\color{blue} \large $ \Sigma_1$}};
\node [above] at (5,2) {{\color{blue} \large $ \Sigma_2$}};

\node [below] at (-5,-2) {{\color{blue} \large $ \Sigma_3$}};
\node [below] at (5,-2) {{\color{blue} \large $ \Sigma_4$}};

\draw[->] (-6,0) -- (6,0) coordinate (x axis);
\draw[->] (0,-3) -- (0,3) coordinate (y axis);

\end{tikzpicture}}
\caption{Contours for the remainder RH problem $R$. \label{Fig5}}
\end{figure}
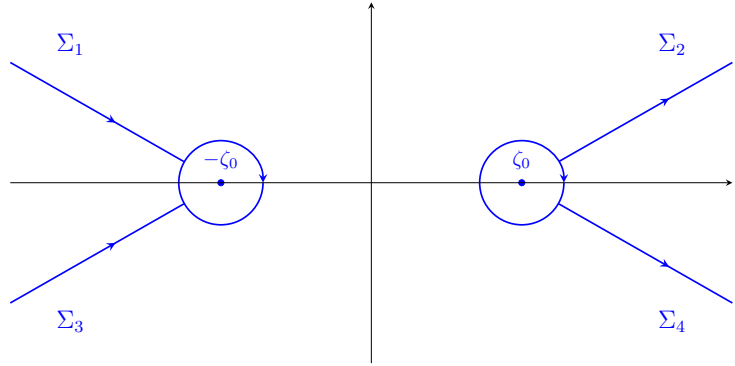

Using the general theory of RH problems normalized at infinity and with jump matrices close to the identity, see e.g. {\cite{DeiftCourant, smallnormRH}}, we can conclude that as $s\to -\infty$, $R$ is close to the identity and has asymptotics of the form
\begin{equation}
R(\zeta)=I+R^{(1)}(\zeta)|s|^{-\frac{2n+1}{2n}}+R^{(2)}(\zeta)|s|^{-\frac{2n+1}{n}}+\mathcal{O}\le( |s|^{- \frac{6n+3}{2n} }\ri), 
\end{equation}
for some matrices $R^{(1)}(\zeta)$ and $R^{(2)}(\zeta)$ which are independent of $s$ and which can be computed via a recursive procedure. We explain how this works for the matrix $R^{(1)}$.
Because of the RH conditions for $R$, we know that $R^{(1)}(\zeta)\to 0$ as $\zeta\to\infty$, and expanding the jump relation for $R$ in $s$, we obtain
\[R_+^{(1)}(\zeta)-R_-^{(1)}(\zeta)=J^{(1)}(\zeta),\qquad \zeta\in\partial U_\pm,\]
whereas $R^{(1)}$ is analytic everywhere else in the complex plane. This allows us to deduce that $R^{(1)}$ is given by
\[
R^{(1)}(\zeta)=\int_{\partial U_+\cup\partial U_-}\frac{J^{(1)}(\xi)}{\xi-\zeta} \frac{\d \xi}{2\pi i}.
\]
This can be computed explicitly using residue calculus.
As $\zeta\to \pm\zeta_0$, there are matrices $A_\pm$, $B_\pm$ independent of $\zeta$ such that we have
\[J^{(1)}(\zeta)=A_\pm(\zeta\mp\zeta_0)^{-2}+B_\pm(\zeta\mp\zeta_0)^{-1}+\mathcal O(1).\]
The entries of $A_\pm$ and $B_\pm$ can be computed explicitly by carefully expanding $J^{(1)}$ as $\zeta\to \pm\zeta_0$.
We have in particular that 
\begin{equation}\label{Bvalues}
B_{\pm,11} = i \lim_{s \rightarrow -\infty}\left(\frac{1}{48 c_1 \zeta_0^2} + \frac{5}{144} \frac{c_2}{c_1^2}\right)=\begin{cases}
 -\frac{i}{32}&\mbox{ if }n=1,\\
 -\frac{i}{8}&\mbox{ if }n>1,
\end{cases}\qquad B_{-,12} = - B_{+,12},\end{equation} 
by \eqref{cvalues} after a straightforward calculation.
It follows that, for $\zeta$ outside $U_\pm$,
$$
R^{(1)}(\zeta)=A_-(\zeta+\zeta_0)^{-2}+B_-(\zeta+\zeta_0)^{-1}+A_+(\zeta-\zeta_0)^{-2}
+B_+(\zeta-\zeta_0)^{-1}.
$$
Expanding this as $\zeta\to\infty$, we get
$$
R^{(1)}(\zeta)=\frac{B_++B_-}{\zeta}+\mathcal O(\zeta^{-2}),
$$
and this implies finally that
\begin{equation}\label{R112}
R_{1}=\big( B_{+}+ B_{-}\big)|s|^{-\frac{2n+1}{2n}}+
M|s|^{-\frac{2n+1}{n}}
+\mathcal{O}\le( |s|^{-  \frac{6n+3}{2n}} \ri),
\end{equation}
as $s\to -\infty$,
where the matrix $M$ can be computed in terms of the jump matrix $R^{(2)}$. The components of $M$ are however unimportant for our concerns.

\subsection{Asymptotics for $q\big((-1)^{n+1}s;1\big)$ and for $\frac{\d}{\d s} \log F(s;1)$ as $s \rightarrow -\infty$.}	

We are now ready to compute the asymptotics for $q\left((-1)^{n+1}s ;1\right)$ as $s\to -\infty$, as stated in \eqref{eq:qas-} and \eqref{eq:qas-2}. Using first \eqref{qdef} and \eqref{S1Psi1}, and then the fact that $S(\zeta)=R(\zeta)P^{\infty}(\zeta)$ for large $\zeta$ together with \eqref{Pinf:as} and the large $\zeta$ asymptotics for $R$, we obtain 
\begin{gather}
q\big((-1)^{n+1}s;1\big)= 2i\Psi_{1,12}(s) =  2i|s|^{\frac{1}{2n}}S_{1,12}(s)=  \zeta_0(s)|s|^{\frac{1}{2n}} + 2i |s|^{\frac{1}{2n}} R_{1,12}.
\end{gather}
By \eqref{zetasolut}, \eqref{Bvalues} and \eqref{R112}, we get
\begin{align}q\big((-1)^{n+1}s;1\big)&= \sum_{i = 0}^{2n+1} \theta_i |s|^{\frac{1}{2n}-\frac{i}{n}}+2i|s|^{-1}\big( B_{+,12}+ B_{-,12}\big)+2iM_{12}|s|^{-2-\frac{1}{2n}}+\mathcal O\left(|s|^{ -2-\frac{3}{2n} }\right)\nonumber\\
&\label{qasformula}= \sum_{i = 0}^{2n} \theta_i |s|^{\frac{1}{2n}-\frac{i}{n}}+(2iM_{12} + \theta_{2n + 1})|s|^{-2-\frac{1}{2n}}+\mathcal O\left(|s|^{-2-\frac{3}{2n}}\right)
\end{align}
as $s\to -\infty$,
where $\theta_0={2n \choose n}^{-\frac{1}{2n}}$ and $\theta_1, \theta_2\ldots$ are as in \eqref{def:thetaj}.
This proves \eqref{eq:qas-} and the improved version \eqref{eq:qas-2}, up to the value of $M_{12}$.

\medskip

To obtain large gap asymptotics, we follow two different routes, and the compatibility between them will allow us to evaluate $M_{12}$. First, we use \eqref{qasformula} to compute 

\begin{equation}q^2\big((-1)^{n+1}s;1\big)=  \sum_{j = 0}^{2n+1} \theta_j^{[2]} |s|^{\frac{1-j}{n}}+(4i\theta_0M_{12} + 2\theta_0\theta_{2n + 1})  |s|^{-2}+{\mathcal O\left(|s|^{-2-\frac{1}{n}}\right)}\end{equation}
as $s\to -\infty$,
where 
\[\theta_j^{[2]}=\sum_{k=0}^j\theta_k\theta_{j-k},\quad j=0,\ldots, n.\] In addition, $\theta_{2n+1}^{[2]}=0$ by Lemma \ref{expansionzeta0}.
Substituting this into \eqref{eq:TWidentity} and integrating, we obtain
\begin{equation}\log F(s;1)=-\sum_{\substack{j=0 \\ j \neq n+1}}^{2n}\frac{n^2}{(n+1-j)(2n+1-j)}\theta_j^{[2]}|s|^{\frac{2n-j+1}{n}} + (4i\theta_0M_{12} + 2\theta_0\theta_{2n + 1})\log|s|+\log C+o(1)\label{largegap1}\end{equation}
as $s\to -\infty$, for some constant $C>0$.

Secondly,
we start from \eqref{eqlogF}, which by \eqref{normalization} and \eqref{S1Psi1} implies
\[\frac{\d}{\d s}\log F(s;1)=2i\Psi_{1,11}=2i|s|^{\frac{1}{2n}}S_{1,11}+2|s|^{\frac{1}{2n}}g_1(s).\]
By \eqref{Pinf:as} and \eqref{defR-}, we get
\[\frac{\d}{\d s}\log F(s;1)=2i|s|^{\frac{1}{2n}}R_{1,11}+2|s|^{\frac{1}{2n}}g_1(s)= \frac{2i\big(B_{+,11}+B_{-,11}\big)}{|s|}+ 2|s|^{\frac{1}{2n}}g_1(s) + \mathcal O\left(|s|^{-\frac{4n+1}{2n}}\right)\]
as $s\to -\infty$.
By \eqref{g_1expansion} and \eqref{Bvalues}, we obtain after integration that there is a constant $C>0$ such that
\[\log F(s;1)=
-\sum_{\substack{j=0 \\ j \neq n+1}}^{2n}\frac{n^2}{(n+1-j)(2n+1-j)}\theta_j^{[2]}|s|^{\frac{2n-j+1}{n}}-\log|s|\ \times \ \begin{cases}\frac{1}{8}&\mbox{if }n=1,\\
\frac{1}{2}&\mbox{if }n>1,\end{cases} +  \log C + o(1).\]
This completes the proof of Theorem \ref{thm:largegap} and Theorem \ref{thm:largegaptau0}.
Comparing this with \eqref{largegap1}, we find
\[4i\theta_0M_{12} + 2\theta_0\theta_{2n + 1}= \begin{cases}-\frac{1}{8}&\mbox{if }n=1,\\
-\frac{1}{2}&\mbox{if }n>1.\end{cases}\] Substituting this in \eqref{qasformula}, we complete the proof of \eqref{eq:qas-2}.

\paragraph{Acknowledgements.}\hfill{}\\
M.C. and M.G. acknowledge the support of the H2020-MSCA-RISE-2017 PROJECT No. 778010 IPaDEGAN. T.C. was supported by the Fonds de la Recherche Scientifique-FNRS under EOS project O013018F.
Part of the work of M.C. and M.G. was done during their visits at the Institut de Recherche en Math\'ematique et Physique (UCLouvain) in May 2018. We acknowledge UCLouvain for excellent working conditions and generous support.
Part of the work was as well conducted during M.C. and M.G.'s visit at Concordia University in July 2018. We thank Concordia University for hosting us and providing numerous resources and facilities for performing our work.
M.C. wants to thank Marco Bertola for many interesting discussions on point processes and integrable operators and, in particular, for teaching him the techniques used in the appendix below to prove that the kernels analyzed in this paper define determinantal point processes.

\appendix
\renewcommand{\theequation}{\Alph{section}.\arabic{equation}}

\section{The generalized Airy kernel point processes}
\label{appendix:genAiry}
 The generalized Airy kernel studied in \cite{FermMM} is given by
\be\label{kernel:genAiryapp}
	K_{\Ai_{2n + 1}}(x,y) := \int_0^\infty \Ai_{2n + 1}(x + u) \Ai_{2n + 1}(y + u)\  \d u,
\ee
where 
\be\label{def:genAiryapp}
\Ai_{2n + 1}(x) = (-1)^n \int_{\gamma_R} \frac{\d \mu}{2 \pi i}\, {\rm e}^{(-1)^{n+1}\frac{\mu^{2n + 1}}{2n +1} - x \mu}=  (-1)^{n+1}\int_{\gamma_L} \frac{\d \lambda}{2 \pi i}\, {\rm e}^{(-1)^n\frac{\lambda^{2n + 1}}{2n + 1} + x \lambda}
\ee
is a real solution of the equation
\be
	\frac{\d^{2n}}{\d x^{2n}} f(x) = (-1)^{n+1} x f (x).
\ee

\begin{proposition}
The kernel \eqref{kernel:genAiryapp} is equal to \eqref{def:kernel} with $p_{2n+1}(\lambda)=\lambda^{2n+1}$.
\end{proposition}
\proof 
By \eqref{def:genAiry}, we have
	\bea
		K_{\Ai_{2n + 1}}(x,y) &=& \int_0^\infty \Ai_{2n + 1}(x + u) \Ai_{2n + 1}(y + u)\ \d u \nonumber \\
		&=& - \frac{1}{(2 \pi i)^2}\int_0^\infty \d u \int_{\gamma_R} \d \mu \int_{\gamma_L} \d\lambda\, {\rm e}^{(-1)^{n+1}\frac{\mu^{2n + 1} - \lambda^{2n + 1}}{2n + 1} - x \mu + y \lambda + u(\lambda - \mu) } \nonumber\\
						&=& \frac{1}{(2 \pi i)^2} \int_{\gamma_R} \d\mu \int_{\gamma_L} \d\lambda \frac{{\rm e}^{(-1)^{n+1}\frac{\mu^{2n + 1} - \lambda^{2n + 1}}{2n + 1} - x \mu + y \lambda}}{\lambda - \mu},
	\eea
	since $\Re\lambda<\Re\mu$ for $\mu\in\gamma_R$ and $\lambda\in\gamma_L$.
\QED

In \cite{FermMM} the authors obtained $K_{\Ai_{2n + 1}}$ as the $N \rightarrow \infty$ scaling limit of the kernels associated to certain point processes describing the momenta of $N$ non--interacting fermions trapped in a non--harmonic potential. In what follows we will prove, without references to any physical model, that indeed, for any $K$ as in \eqref{def:kernel} and $p_{2n+1}$ as in \eqref{def p} , the determinants 
\be\label{corrfcts}
	\rho_\ell(x_1,\ldots,x_\ell) := \det \Big[K(x_i,x_j) \Big]_{i,j = 1}^\ell, \; \ell \geq 1
\ee  
are the correlation functions associated to a well defined point process. The proof {uses} Theorem 3 in \cite{Sosh}, stating that {a Hermitian locally trace class} operator $\mathcal K$ uniquely defines a determinantal point process if and only if it is positive definite and bounded from above by the identity {operator}.

\begin{theorem}
 	For any $n \geq 1$, there exists a unique determinantal point process whose correlation functions are given by \eqref{corrfcts}.
\end{theorem}
\proof

For future use, we start {by} rewriting the kernel $K$ as 
\be\label{kernelrewritten}
	K(x,y) = \int_{0}^\infty \varphi (x + u) \varphi(y + u) \d u,
\ee
where 
\be\label{def varphi}
	\varphi(x) := \frac{1}{2 \pi} \int_{-\infty}^{+ \infty} {\rm e}^{i \left(\frac{\lambda^{2n + 1}}{2n + 1} +\sum_{ j = 1}^{n-1} \frac{(-1)^{n + j}\tau_j}{2j + 1} \lambda^{2j + 1} + x\lambda \right)}\d \lambda.
\ee

From \eqref{kernelrewritten} we immediately deduce that the  integral operator $\mathcal K$ with kernel $K$ is Hermitian (since $\varphi$ is real) and locally of trace class (because of continuity). In order to prove total positivity, we use the { Andr\'eief identity} 

\begin{gather}
\det \le[\int_{[0,\infty)} \phi_i(x) \psi_j(x)\, \d x  \ri]_{i,j=1}^\ell = \frac{1}{\ell!} \int_{[0,\infty)^\ell} \det\big[ \phi_i(x_j)\big]_{i,j=1}^\ell \det\big[\psi_i(x_j)\big]_{i,j=1}^\ell \, \d x_1\cdots \d x_\ell
\end{gather}

giving, in our case,
\begin{gather}
\det \Big[ K(x_i, x_j) \Big]_{i,j=1}^\ell = \det \le[  \int_0^\infty \varphi(x_i + u) \varphi(x_j + u)\  \d u \ri]_{i,j =1}^\ell \nonumber \\
= \frac{1}{\ell!} \int_{(0,+\infty)^\ell} \det\big[ \varphi(x_i + u_s)\big]_{i,s=1}^\ell \det\big[\varphi(x_j + u_t)\big]_{j,t=1}^\ell \, \d u_1\cdots d u_\ell \nonumber \\
=  \frac{1}{\ell!} \int_{(0,+\infty)^\ell} \le( \det\big[ \varphi(x_i + u_s)\big]_{i,s=1}^\ell \ri)^2  \, \d u_1 \cdots d u_\ell  \geq 0.
\end{gather}

Now, in order to prove that $\mathcal K$ is bounded from above by the identity, we introduce the associated (generalized) Airy transform as
\be\label{Airytransform}
	\Phi = \mathfrak F^{-1} \mathcal M \mathfrak F^{-1},
\ee
where $\mathfrak F$ is the standard Fourier transform such that
$$
	(\mathfrak F g)(x) = \frac{1}{\sqrt{2 \pi}} \int_{-\infty}^\infty g(\lambda){\rm e}^{-i x \lambda} \d t
$$
and $\mathcal M$ is {the multiplication operator} by the function ${\rm e}^{i\Big(\frac{\lambda^{2n + 1}}{2n + 1} + \sum_{ j = 1}^{n-1} \frac{(-1)^{n + j}\tau_j}{2j + 1} \lambda^{2j + 1}\Big)}$. This is a straightforward {generalization of what was }done in \cite{BW} for the Airy case $n = 1$. Note that, in particular, 
$$
	(\Phi f)(x) = \int_{-\infty}^\infty \varphi(x + u)f(u) \d u.
$$
 Using the equation above in combination with the definition \eqref{kernelrewritten} of $K$ we obtain
\be
	\mathcal K = \Phi \mathcal P \Phi,
\ee 
 where $\mathcal P$ is the multiplication operator for the characteristic function $\chi_{[0,\infty)}$. We will use the equation above to prove that $\mathcal K$ is a projection operator, i.e. 
$\mathcal K^2 = \mathcal K$, which implies that   $\mathcal  K \leq 1 $. Indeed, denote with $\mathcal I$ the inversion operator such that
$$(\mathcal{I} g)(x) = g(-x).$$ 
Since
$$\mathfrak F^{-1} = \mathcal I \mathfrak F = \mathfrak F \mathcal I,$$ 
we immediately see that
$$\Phi = \mathfrak F^{-1}  \mathcal M \mathfrak F^{-1} = \mathfrak F \mathcal  {I M I} \mathfrak F = \mathfrak F  \mathcal M^{-1} \mathfrak F = \Phi^{-1} $$
from which the reproducing property $\mathcal K^2 = \mathcal K$ follows immediately. \QED

\bibliographystyle{plain}
\bibliography{FermionsPIIhier_biblio.bib}
\end{document}